\documentclass[11pt,fleqn,a4paper]{article}
\usepackage{graphicx}
\usepackage{fullpage}
\usepackage{amssymb}
\usepackage{amsthm}
\usepackage{epstopdf}
\usepackage[flushleft]{threeparttable}
\usepackage{pdflscape}
\usepackage{amsmath}
\usepackage{tkz-graph}
\usetikzlibrary{shapes.geometric}
\usepackage{subfigure}
\usepackage{setspace}
\usepackage[bottom]{footmisc}
\usepackage{endnotes}
\usepackage{graphicx}
\usepackage{textcomp}
\usepackage{adjustbox} 
\usepackage{natbib}
\usepackage[ansinew]{inputenc}
\usepackage{verbatim}
\usepackage{lscape}
\usepackage{array}
\usepackage{fltpoint}
\usepackage{dcolumn}
\usepackage{booktabs}
\usepackage{rotating}
\usepackage[norounding,point]{rccol}
\usepackage{hyperref}
\usepackage{placeins}
\usepackage{multirow}
\usepackage{paralist}
\usepackage{enumitem}
\usepackage{subcaption}

\newtheorem{assumptionA}{Assumption}
\newtheorem{assumptionB}{Assumption}

\usepackage{wasysym}
\usepackage{mathtools}
\usepackage{marvosym}
\usepackage{amsthm}

\usepackage[margin=0.8in]{geometry}

\usepackage{tikz}

\usetikzlibrary{shapes,decorations,arrows,calc,arrows.meta,fit,positioning}
\tikzset{
    -Latex,auto,node distance =1 cm and 1 cm,semithick,
    state/.style ={ellipse, draw, minimum width = 0.7 cm},
    point/.style = {circle, draw, inner sep=0.04cm,fill,node contents={}},
    bidirected/.style={Latex-Latex,dashed},
    el/.style = {inner sep=2pt, align=left, sloped}
}

\numberwithin{equation}{section}

\hypersetup{
colorlinks=true,
urlcolor=blue,
linkcolor=red,
citecolor=blue
}

\usepackage{setspace}%%
\usepackage{xspace}%% für besseres Handling der spaces

\newtheorem{theorem}{Theorem}

\newtheorem{assumption}{Assumption}

\tikzstyle{VertexStyle} = [shape            = ellipse,
minimum width    = 6ex,%
draw]

\tikzstyle{EdgeStyle}   = [->,>=stealth']

\begin{document}
\sloppy

\begin{center}

{\LARGE Testing Full Mediation of Treatment Effects\\and the Identifiability of Causal Mechanisms}

{\large \vspace{0.1cm}}

{\large Martin Huber*, Kevin Kloiber+, Luk\'{a}\v{s} Laff\'{e}rs**} \vspace{0.1cm}

{\footnotesize{*Department of Economics, University of Fribourg, Fribourg, Switzerland \\ +Department of Economics, University of Munich, Munich, Germany\\  **Department of Mathematics, Matej Bel University, Bansk\'a Bystrica, Slovakia\\ **Department of Economics, Norwegian School of Economics, Bergen, Norway}}
\end{center}

\thispagestyle{empty}
%\vspace*{0.5 cm}
\begin{abstract}
  {\footnotesize In causal analysis, understanding the causal mechanisms through which an intervention or treatment affects an outcome is often of central interest. We propose a test to jointly evaluate (i) whether the causal effect of a treatment that is randomly assigned conditional on covariates is fully mediated by, or operates exclusively through, observed intermediate outcomes (referred to as mediators or surrogate outcomes), and (ii) whether the various causal mechanisms operating through different mediators are identifiable conditional on covariates. We demonstrate that if full mediation and identification of causal mechanisms hold jointly, then the conditionally random treatment is conditionally independent of the outcome given the mediators and covariates. Furthermore, we extend our framework to settings with non-randomly assigned treatments. We show that, in this case, full mediation remains testable, while identification of causal mechanisms is no longer guaranteed. We propose a double machine learning framework for implementing the test that can incorporate high-dimensional covariates and is root-n consistent and asymptotically normal under specific regularity conditions. We also present a simulation study demonstrating good finite-sample performance of our method, along with two empirical applications revisiting randomized experiments on maternal mental health and social norms.}
  \\[0.5cm]\it{JEL Classification:} C12, C21 \\[0.1cm]
{\it Keywords: mediation, causal mechanisms, double machine learning, conditional independence, treatment effects, hypothesis test.}
\end{abstract}

\begin{spacing}{1} 
{\scriptsize  We have benefited from comments by Jonathan Roth. %of seminar participants , as well as conference participants at . 
Addresses for correspondence: Martin Huber, University of Fribourg, Bd.\ de P\'{e}rolles 90, 1700 Fribourg, Switzerland; martin.huber@unifr.ch. Kevin Kloiber, University of Munich, Ludwigstrasse 33, 80539 Munich, Germany, and ifo Institute, Poschingerstrasse 5, 81679 Munich, Germany; kevin.kloiber@econ.lmu.de. Luk\'{a}\v{s} Laff\'{e}rs, Department of Mathematics, Matej Bel University, Tajovsk\'{e}ho 40, 974 01 Bansk\'{a} Bystrica, Slovakia; Department of Economics, Norwegian School of Economics, Helleveien 30, 5045 Bergen, Norway; lukas.laffers@gmail.com. Laff\'{e}rs acknowledges support  provided by the Slovak Research and Development Agency under contracts VEGA 1/0645/26 and APVV-21-0360.}
\end{spacing}
%We introduce a method to test the identification of causal effects in causal mediation and dynamic treatment models, accounting for observed covariates and suspected instruments. Our approach builds upon \citet{huberkueck2022}'s test designed for single treatment models. The models under study involve sequential assignment of treatment and mediator to assess direct treatment effects (net of the mediator), indirect effects via the mediator, or joint effects of both treatment and mediator.
%
%We establish testable conditions sufficient for identifying these effects in observational data. These conditions jointly imply: (1) the exogeneity of treatment and mediator given covariates, and (2) the validity of at least two distinct instruments for treatment and mediator. These instruments must not directly affect the outcome, other than through the treatment or mediator, and should be unconfounded given covariates.
%
%Our method extends to scenarios involving sample selection or post-treatment attrition issues. In such cases, the mediator is replaced by a selection indicator for observed outcomes, enabling joint testing of treatment selectivity and attrition. We propose machine learning-based tests that incorporate covariates in a data-driven manner and evaluate their finite sample performance through simulations.
%
%Additionally, we apply our testing method to labor market data from Slovakia. Our findings suggest that the testable implications crucial for identification remain unchallenged for certain sequentially assigned training programs.

\thispagestyle{empty}
\newpage
\setcounter{page}{1}

\section{Introduction} 

In many causal studies, interest extends beyond estimating the total effect of an intervention or treatment (such as education) on an outcome (such as health) to understanding the causal mechanisms through which the effect operates, such as hypothesized peer effects in schools or potential education-induced changes in health behavior. Causal mediation analysis aims to disentangle the indirect effect of a treatment, mediated by one or more intermediate variables (mediators), from the direct effect of the treatment on the outcome that does not operate through these mediators. In particular, researchers may be interested in determining whether a direct effect exists at all after accounting for observed mediators. The absence of a direct effect implies that the treatment effect is fully mediated.

A common approach in empirical studies to test for the absence of a direct effect is to estimate an outcome regression on the treatment and assess whether the treatment coefficient becomes approximately zero (and statistically insignificant) once the mediators are included as explanatory variables. However, this test is generally invalid when mediators are endogenous, that is, when unobserved confounders jointly affect both the mediator and the outcome. In such cases, controlling for mediators can introduce spurious associations between the treatment and the outcome due to collider bias or post-treatment selection bias, as for instance discussed in \citet{RoGr92}. This problem, which motivates an alternative testing approach proposed by \citet{kwon2024testing} and outlined further below, arises even when the treatment is randomly assigned and leads to biased estimation of the direct effect. As also noted by \citet{acharya2016explaining}, naively controlling for endogenous mediators does not permit valid inference about the presence or absence of direct effects unless both the treatment and the mediators are exogenous, at least conditional on observed covariates.

While testing whether the treatment effect is zero after including mediators cannot be used to establish full mediation based on the treatment effect alone, this paper demonstrates that it can be used to jointly test for full mediation and specific exogeneity conditions. More precisely, we consider a causal model with a treatment that is randomly assigned, at least conditional on observed covariates. We show that tests of the conditional independence between treatment and outcome, conditional on the mediators and possibly additional covariates to control for mediator-outcome confounding, permit the joint testing of two conditions: (i) the absence of a direct treatment effect on the outcome (i.e., full mediation), and (ii) the identifiability of the indirect effects operating through the mediators. 

It is important to emphasize that our test is informative only about factual outcomes - those observed under the mediator values actually realized in the data. For these factual outcomes, satisfaction of the conditional independence between treatment and outcome is necessary and sufficient for the joint satisfaction of full mediation and identifiability. This means that while we can detect violations of our assumptions in observed data, our test cannot detect violations that might occur exclusively in counterfactual outcomes, such as the outcome of a mediated unit had it not been mediated. While this represents a theoretical boundary, 'counterfactual-only' violations appear unlikely in most empirical settings, as they would require highly restrictive and specific data-generating mechanisms to avoid manifesting in any of the factual outcomes as well.

We also extend the framework to non-randomly assigned treatments. We demonstrate that, in this case, our conditional independence test still permits verification of full mediation. However, it no longer permits testing the identifiability of causal mechanisms. The reason is that while our test is informative about mediator-outcome confounding, it is not informative about treatment-mediator confounding (conditional on covariates). The latter must be ruled out by assumption (and can be ruled out by design under treatment randomization) for the identification of causal mechanisms.

We also link our results to two identification strategies in observational studies for assessing causal effects: (i) the back-door (BD), or selection-on-observables, strategy, see for instance the survey by \citet{Im04}, which assumes the absence of treatment-outcome confounders conditional on covariates; and (ii) the front-door (FD) approach, see e.g. \citet{Pearl00} and \citet{Fulcheretal2019}, which assumes full mediation as well as the absence of treatment-mediator and mediator-outcome confounders conditional on covariates. We show that satisfaction of our testable condition - conditional independence between the treatment and the outcome given the mediators and covariates - implies that BD and FD approaches targeting features of the potential outcome distributions conditional on covariates, such as the conditional average treatment effect (CATE), yield identical parameters. 

We also provide a sufficient condition under which the converse implication holds, namely a separability assumption. The latter assumes that the conditional probability of the outcome given treatment, mediator, and covariates can be modeled by two additive functions: one that depends on the treatment and covariates, and another that depends on the mediator and covariates. This strong parametric condition, however, is unlikely to be satisfied in many empirically relevant settings, as it rules out treatment-mediator interaction effects on the outcome conditional on covariates and thus severely restricts effect heterogeneity \citep{RoGr92,Robins2003}. Therefore, considering our testable implication appears more attractive than running a BD-FD comparison, which, in the absence of separability, may yield equivalence even if full mediation (or other identifying assumptions) fails.

For the practical implementation of the test of conditional independence of the treatment and outcome, we rely on double machine learning (DML) methods as proposed in \citet{Chetal2018}. These methods employ doubly robust (DR) score functions for causal parameter estimation and enable flexible adjustment for high-dimensional covariates using machine learning, while maintaining root-\(n\) consistency under suitable regularity conditions. A key property of DR scores is \citet{Neyman1959} orthogonality, which makes DR estimation relatively insensitive to biases in the estimation of nuisance components such as outcome, treatment, or mediator models, such as regularization bias common in machine learning methods. 

More specifically, we build on \citet{Apfel2023learning}, who consider a closely related testing problem of jointly verifying instrument validity and treatment exogeneity given observed covariates, and propose a DR score for a joint conditional independence test across all values of the (possibly high-dimensional) conditioning set. In the same vein (and even if practically less attractive due to the required separability condition), we also propose a DML test comparing BD and FD methods, where the test statistic is based on the mean difference of the DR score functions for conditional mean potential outcomes under each approach. See \cite{Robins+94} and \cite{RoRo95} for BD-based DR score functions under the selection-on-observables assumption, and \cite{Fulcheretal2019} for analogous DR scores based on the FD criterion.

We investigate the finite-sample performance of our conditional independence test in a simulation study and find it to perform well for sample sizes of a few thousand observations. As an empirical illustration, we apply our method to two randomized experiments previously considered by \citet{kwon2024testing} (for testing full mediation alone): a cognitive-behavioral therapy intervention for maternal depression in Pakistan \citep{baranov2020} and an information intervention on social norms regarding women's employment in Saudi Arabia \citep{bursztyn2020}. In both cases, we reject the joint null of full mediation and mediator exogeneity.

Our paper is related to a large literature on causal mediation analysis that assesses direct and indirect effects based on selection-on-observables assumptions imposing the conditional exogeneity of the treatment and the mediator given observed covariates, see for instance \citet{FlFl09}, \citet{ImKeYa10}, \citet{Hong10}, \citet{TchetgenTchetgenShpitser2011}, and \citet{Huber2012}, among many others. It is specifically related to studies aiming at testing the absence of a direct effect or full mediation when (implicitly or explicitly) imposing such selection-on-observables assumptions. In contrast to this literature, we do not impose selection-on-observables assumptions a priori, but propose a method that jointly tests this assumption together with full mediation under specific conditions. 

Our work is also closely related to \citet{deLunaJohansson2012} and \citet{huberkueck2022}, who develop tests for the conditional independence of an instrumental variable given treatment and covariates. Under certain conditions, this implies that the selection-on-observables assumption holds for the treatment and the instrument, and that the instrument does not directly affect the outcome. In our setting, the treatment plays the role of the instrument, and the mediators take the role of the treatment in the context of their tests. As noted above, we make use of a DR approach proposed in \citet{Apfel2023learning}, which yields an asymptotically normal test statistic under suitable regularity conditions. 

Our paper is also related to \citet{kwon2024testing}, who develop a test for full mediation that is robust to mediator endogeneity and therefore does not imply simultaneously verifying the identifiability of indirect effects based on the selection-on-observables assumption. Their method is based on a partial identification approach akin to the literature on testing instrumental variable (IV) assumptions; see, for example, \citet{Kitagawa2008}, \citet{HuMe11}, \citet{MoWa2014}, \citet{farbmacher2022instrument}, and \citet{sun2023instrument}.  \citet{kwon2024testing} show that the testable constraints on the joint distribution of the outcome and mediator conditional on treatment, arising from the multiple mediator values considered in their paper, are sharp, in the sense that they exhaust all testable information in the data: if these constraints are satisfied, there exists a data generating process that satisfies full mediation and is consistent with the observable data.

Relative to \citet{kwon2024testing}, one may ask why it is valuable to test the joint null of full mediation and mediator exogeneity, given that these correspond to distinct types of failure -- one concerning the causal mechanism and the other statistical identification. Our joint approach is motivated by the fact that, in many empirical applications of causal mediation analysis, the objective is not merely to demonstrate that a mechanism may exist, but to point-identify its magnitude. Because point identification of all mechanisms through which a treatment affects the outcome, based solely on observed mediators, is not possible if either full mediation fails or the mediator is endogenous, rejection of our joint null provides a critical diagnostic signal. While \citet{kwon2024testing} focus on whether there exists any distribution of potential outcomes consistent with full mediation, even in the presence of endogeneity, our test is additionally useful for assessing whether fully mediated treatment effects are point-identified.

Beyond this conceptual difference, our framework also differs from \citet{kwon2024testing} along several methodological dimensions. First, we
do not impose monotonicity of the mediator in the treatment, which broadens the range of empirical applications in which the test can be applied.
Second, our procedure accommodates high-dimensional covariates \( X \) and is compatible with flexible machine-learning estimators of the relevant nuisance functions. Third, we allow the mediator to be multivariate and continuous, rather than restricting attention to discrete mediators. Finally, we provide results for settings in which the treatment \( D \) is not randomly assigned, so that the test remains informative beyond randomized controlled trials.

Our study is also related to the literature on causal analysis based on so-called surrogate outcomes, as discussed in \cite*{frangakis2002principal} and \citet{athey2020estimating}. This framework is particularly useful when the primary outcome of interest is not observed in the dataset in which the treatment is measured. Instead, researchers may observe one or more surrogate outcomes, through which the treatment effect on the primary outcome is assumed to operate entirely - that is, under the assumption of full mediation via the surrogate outcomes, which then function as mediators. \citet{prentice1989surrogate} notes that, under full mediation, the treatment and the (unobserved) primary outcome are conditionally independent given the surrogate outcome.  However, as we formalize in our paper using the potential outcomes framework, this also requires that the surrogate outcome is conditionally exogenous and that the treatment is randomly assigned (at least conditional on covariates).

The remainder of this paper is organized as follows. Section \ref{Assumptions1} introduces the causal effects of interest and the identifying assumptions under a randomized treatment (at least conditional on covariates), which allow for testing full mediation and the identifiability of the indirect effects. It also outlines the testable conditional independence condition. Section \ref{Assumptions2} discusses the identifying assumptions when the treatment is not (conditionally) randomized, demonstrating that the conditional independence condition can then only be used to test for full mediation, but not for the identifiability of the indirect effects. Section \ref{Assumptions3} relates our testable implication to the testing power implied by comparing parameters obtained from the BD and FD criteria. Section \ref{testing} describes the empirical implementation of our tests based on DML using DR score functions. Section \ref{simulations} presents a simulation study evaluating the finite-sample performance of the proposed tests. Section \ref{application} applies our methodology to the experimental studies of \citet{baranov2020} and \citet{bursztyn2020}. Section \ref{conclusion} concludes.

\section{Causal framework for (conditionally) random treatments}\label{Assumptions1}

Causal mediation analysis seeks to decompose the overall causal effect of a treatment variable, denoted by \( D \), on an outcome variable, \( Y \), into distinct causal mechanisms. Specifically, the effect of \( D \) on \( Y \) may be mediated through one or more observed mediators, denoted by \( M \), while a direct effect of \( D \) on \( Y \) may also exist through other (potentially unobserved) pathways.  Throughout this discussion, we denote random variables in capital letters and their realized values in lowercase letters. To formalize causal mechanisms, we adopt the potential outcomes framework, as, for example, considered in \citet{Neyman23} and \citet{Rubin74}. Let \( M(d) \) represent the potential mediator under treatment level \( d \in \mathcal{D} \), where \( \mathcal{D} \) is the support of \( D \). Similarly, let \( Y(d,m) \) denote the potential outcome under treatment \( d \) and mediator value \( m \in \mathcal{M} \), where \( \mathcal{M} \) is the support of \( M \).  For each individual, we define \( Y(d,m) \) and \( M(d) \) as functions of their assigned treatment \( D = d \) and mediator value \( M = m \), assuming:  
(i) an individual's potential outcomes are unaffected by the treatment or mediator status of others, and  
(ii) there are no multiple versions of any treatment or mediator level across individuals.  These conditions constitute the Stable Unit Treatment Value Assumption (SUTVA), see e.g.\ \citet{Cox58} and \citet{Rubin80}. Additionally, we denote observed pre-treatment covariates as \( X \), with support \( \mathcal{X} \). We denote the support of $Y$ by \( \mathcal{Y} \). Throughout the paper, $A \bot B \mid C$ denotes statistical independence of $A$ and $B$ conditional on $C$; this is a property of the joint distribution of the variables involved and does not, by itself, encode a causal relationship. Under the causal faithfulness imposed in Assumption~\ref{A1} below, however, statistical independence is equivalent to d-separation in the corresponding causal graph.

We start by defining causal effects that are of interest in the context of treatment evaluation and mediation analysis. For simplicity, we focus on a binary treatment and on average effects (which are the primary focus in applied work) in the subsequent discussion. However, the assumptions and testable implications we present further below more generally refer to the distribution of potential outcomes and also apply to non-binary treatments. A first parameter of interest is the average treatment effect (ATE), denoted by \( \Delta \), which comprises both direct and indirect effects of the treatment:  
\begin{align}
    \Delta = E[Y(1,M(1)) - Y(0,M(0))].
\end{align}  
Causal mediation analysis aims to disentangle the causal mechanisms underlying the ATE; see, e.g., \citet{Robins2003} and \citet{Pearl01}.  
By varying the treatment while holding the mediator fixed at some value \( M = m \), we obtain the average controlled direct effect (CDE), denoted by \( \theta(m) \):  
\begin{align}
\theta(m) = E[Y(1,m) - Y(0,m)], \quad m \in \mathcal{M}.
\end{align}  
By varying the treatment while holding the mediator fixed at its potential value under treatment \( D = d \), we obtain the average natural direct effect (NDE), denoted by \( \theta(M(d)) \):  
\begin{align}
\theta(M(d)) = E[Y(1,M(d)) - Y(0,M(d))], \quad d \in \{0,1\}.
\end{align}  
Conversely, by varying the mediator between its potential values under treatment and nontreatment while keeping the treatment fixed at \( D = d \), we obtain the average natural indirect effect (NIE):  
\begin{align}
\delta(d) = E[Y(d,M(1)) - Y(d,M(0))], \quad d \in \{0,1\}.
\end{align}  
Notably, the ATE can be decomposed into the sum of the NDE and NIE evaluated at opposite treatment states:  
\begin{align}\label{ate}
\Delta = \theta(1) + \delta(0) = \theta(0) + \delta(1).
\end{align}  
Furthermore, if the ATE of $D$ on $Y$ is fully mediated by $M$ such that the mean potential outcome depends only on the mediator, i.e., \( E[Y(d,m)] = E[Y(m)] \), then any direct effect like the CDE or NDE is zero, implying \( \theta(m) = \theta(M(d)) = 0 \). Under this condition, equation \eqref{ate} simplifies to:  
\begin{align}\label{ate_indirect}
\Delta = \delta(0) = \delta(1) = \delta = E[Y(M(1)) - Y(M(0))].
\end{align}  
Thus, the ATE corresponds entirely to the indirect effect.   

Assuming that treatment is randomly assigned, at least when controlling for covariates $X$, we propose a joint test for (i) the absence of a direct effect and (ii) the exogeneity of the mediator \( M \). If both assumptions hold, we can identify the effects of different mediator values,  
\begin{align}
E[Y(m) - Y(m')], \quad \text{for } m \neq m' \text{ and } m,m' \in \mathcal{M},
\end{align}  
in addition to the ATE (or NIE) related to a treatment-induced shift in the mediator provided in equation \eqref{ate_indirect}.  
If the test rejects asymptotically, then either $M$ is not exogenous, $D$ has a direct effect on $Y$ (i.e., $\theta(m)\neq 0$ for some $m \in \mathcal{M}$), or both.  While we have so far considered a binary treatment for expository purposes, our test also applies to nonbinary treatments. In this case, the ATE and the indirect effect, if it fully explains the ATE for two distinct treatment levels \( d, d' \) is given by:  
\begin{align}
\Delta_{d,d'} = E[Y(M(d)) - Y(M(d'))].
\end{align}  

After defining the causal effects, we subsequently discuss several identifying assumptions required for testing. Our first assumption establishes a specific causal structure between variables within our framework, positing that only certain variables exert a causal influence on others. We formalize this causal structure using the previously mentioned potential outcome notation, by applying it to variables beyond the outcome and the mediator. Let \(A(b)\) and \(A(b,c,\ldots)\) denote the potential value of variable \(A\) when variable \(B\) is set to \(b\), or when variables \(B\), \(C\), etc., are set to \(b\), \(c\), and so on. Furthermore, we assert that if a causal relationship exists between two variables (potentially conditional on other variables), then a statistical dependence must exist between them, consistent with the principle of causal faithfulness. We work within the SWIG template framework of \citet{richardson2013single} under the FFRCISTG semantics of \citet{Ro86} and \citet{RobinsRichardson2010}: all interventional graphs are generated from a common template, but no cross-world independence (in the NPSEM-IE sense) is imposed.

\begin{assumption}[Causal structure and faithfulness]\label{A1}  
\begin{align*}  
  M(y) = M, \quad D(m, y) = D, \quad X(d, m, y) = X, \\  
  \forall d \in \mathcal{D}, m \in \mathcal{M}, y \in \mathcal{Y},  
\end{align*}  
only variables that are d-separated in any causal model are statistically independent.  
\end{assumption}  
\noindent The first line of Assumption~\ref{A1} excludes any reverse causal effect of outcome \( Y \) on \( D \), \( X \), or \( M \). Additionally, the treatment \( D \) must not causally affect \( X \), and the mediator \( M \) must not causally affect \( X \) or \( D \), which holds if \( X \) represents pre-treatment covariates. Assumption~\ref{A1} also establishes causal faithfulness, which ensures that only variables that are d-separated, i.e., not linked via any causal paths, are statistically independent or conditionally independent (see, e.g., \citep{Pearl00}). 

To be more precise, we employ the d-separation criterion of \citet{pearl1988probabilistic}, which relies on blocking causal paths between variables. A path between two sets of variables, \( A \) and \( B \), is blocked when conditioning on a set of control variables, \( C \), if:  
\begin{enumerate}  
	\item the path between \( A \) and \( B \) is a causal chain, implying that \( A\rightarrow M \rightarrow B \) or \( A\leftarrow M \leftarrow B \), or a confounding association, implying that \( A\leftarrow M \rightarrow B \), and variable (set) \( M \) is among control variables \( C \) (i.e., controlled for),  
	\item the path between \( A \) and \( B \) contains a collider, implying that \( A\rightarrow  S  \leftarrow B \), and variable (set) \( S \) or any variable (set) causally affected by \( S \) is not among control variables \( C \) (i.e., not controlled for).  
\end{enumerate}  
According to the d-separation criterion, \( A \) and \( B \) are d-separated when conditioning on a set of control variables \( C \) if, and only if, \( C \) blocks every path between \( A \) and \( B \). D-separation is sufficient for the (conditional) independence of two variables, providing a crucial component for the proof of our theorems. Causal faithfulness further imposes that d-separation is also a necessary condition, such that two variables are statistically independent if and only if d-separation holds.  A scenario in which faithfulness fails occurs when one variable affects another via multiple causal paths (or mechanisms) that exactly cancel out, such that the variables appear independent despite an underlying causal relationship (see, e.g., the discussion in \citet{spirtes2000causation}).

Next, we introduce a common support assumption that is required for nonparametric testing. To this end, let \( f(A=a \mid B=b) \) denote the conditional density of variable \( A \) given \( B \) at values \( A=a \) and \( B=b \). If \( A \) is discrete rather than continuous, then \( f(A=a \mid B=b) \) represents a conditional probability rather than a density.  
\begin{assumption}[Common support]\label{A2}  
  \begin{align*}  
  f(D=d,  M=m \mid X=x) > 0 \quad \forall d \in \mathcal{D}, m \in \mathcal{M}, \textrm{ and } x \in \mathcal{X}.  
  \end{align*}  
\end{assumption}  
\noindent Assumption~\ref{A2} requires that, conditional on \( M \) and \( X \), observations exist for all values of \( D \) that occur in the total population. If this condition is violated, total mediation and the identifiability of effects of \( m \) can only be tested over a subset of the support of \( X \), \( D \), and/or \( M \).  
Assumption~\ref{A2} further implies both  $f(D=d \mid M=m, X=x) > 0$ and $f(M=m \mid D=d, X=x) > 0$. Satisfaction of \( f(D=d \mid M=m, X=x) > 0 \) ensures that, for any values of mediators and covariates in the population, we have common support for testing whether the effects of \( D \) on \( Y \) are fully mediated by \( M \). Similarly, \( f(M=m \mid D=d, X=x) > 0 \) implies that, for any values of treatment and covariates, common support exists for testing whether the effect of \( M \) on \( Y \) is unconfounded.  
A violation of common support entails that one may jointly test full mediation and mediator exogeneity only for a subset of covariate, treatment, or mediator value combinations. Thus, Assumption~\ref{A2} is not strictly necessary for implementing our testing approach if we restrict attention to a subpopulation satisfying common support. However, this restriction to a subpopulation may reduce testing power.  

Our next assumption imposes that the treatment has a first-stage effect on the mediator:  
\begin{assumption}[Treatment effect on mediator]\label{A3}  
  \begin{align*}  
  \Pr(M(d) \neq M(d') \mid X = x) > 0 \quad \text{for some } d, d' \in \mathcal{D} \text{ and } \forall x \in \mathcal{X}.  
  \end{align*}  
\end{assumption}  
\noindent This assumption ensures that treatment induces variation in the mediator for at least some individuals. If it fails, the indirect effect is necessarily zero, as \( M \) does not respond to changes in \( D \). If treatment does not influence the mediator, mediation analysis becomes meaningless.  

Next, we assume that the treatment is as good as randomly assigned, as in an experiment, at least within strata defined by covariates \( X \) (or subsets thereof). This implies the following conditional independence assumption:
\begin{assumption}[Conditional independence of the treatment]\label{A4}  
  \begin{align*}  
   \{Y(d,m), M(d)\}  \bot D \mid X = x \quad \forall d \in \mathcal{D}, m \in \mathcal{M}, \text{ and } x \in \mathcal{X}.  
  \end{align*}  
\end{assumption}  
\noindent This assumption states that, conditional on \( X \), there are no unobserved confounders jointly affecting the treatment \( D \) and any
post-treatment variable, namely the mediator \( M \) or the outcome \( Y \). This is known as the exogeneity, unconfoundedness, exchangeability, ignorability, or selection-on-observables assumption. See, for instance, \citet{Im04} for a survey. We state Assumption~\ref{A4} in the joint form \(\{Y(d,m), M(d)\} \perp D \mid X = x\) for all \(d \in \mathcal{D}\), \(m \in \mathcal{M}\), and \(x \in \mathcal{X}\) because this is the form that is naturally delivered by design. Whenever \( D \) is (conditionally) randomized, as in a randomized controlled trial or a setting in which treatment is as good as random given \( X \), the full vector of counterfactuals \(\{Y(d,m), M(d)\}\) is independent of \( D \) given \( X \). In the terminology of \citet{Imbens99}, this corresponds to \emph{strong} unconfoundedness. The weaker variant requires only marginal independence between \( D \) and each potential outcome separately, and is what is actually needed for identification of average effects. Subsequent results in the paper exploit only this weaker, single-world content. Assumption~\ref{A4} is trivially satisfied in experiments in which the treatment is randomized unconditionally, such that controlling for \( X \) is not required for identification. Under Assumption~\ref{A4}, Assumption~\ref{A3} becomes testable by examining whether \( M \) varies across values of \( D \), conditional on \( X \).

Maintaining Assumptions \ref{A1} to \ref{A4}, we aim to test the following two assumptions. The first imposes full mediation, implying that the average treatment effect (ATE) corresponds to the indirect effect. The second imposes selection on observables with respect to the mediator.  
\begin{assumption}[Conditional full mediation]\label{Afullmed}  
  \begin{align*}  
   \Pr(Y(d,m) =Y(m) \mid X=x) = 1 \quad \forall d \in \mathcal{D}, m \in \mathcal{M}, \text{ and } x  \in \mathcal{X}.  
  \end{align*}  
\end{assumption}  
\noindent Assumption \ref{Afullmed} states that any effect of treatment \( D \) on \( Y \) is fully mediated by \( M \), conditional on \( X \). If \( X \) consists only of pre-treatment covariates, this implies the stronger version \( \Pr(Y(d,m) =Y(m)) = 1 \), since \( D \) and \( M \) cannot affect \( Y \) through variables that causally precede them. If \( X \) includes post-treatment variables, this assumption allows \( D \) to affect \( Y \) via \( X \), but stipulates that the only other causal mechanism through which \( D \) influences \( Y \) is the mediator \( M \). In this case, testing focuses exclusively on the indirect effect of \( D \) through \( M \), and does not consider other potential causal mechanisms via \( X \) that may exist in addition to those through \( M \).

\begin{assumption}[Conditional independence of the mediator]\label{Acimed}  
  \begin{align*}  
   Y(m)  \bot M \mid X = x \quad \forall m \in \mathcal{M}, \text{ and } x \in \mathcal{X}.  
  \end{align*}  
\end{assumption}  
\noindent Assumption \ref{Acimed} (Mediator exogeneity) implies that, given the covariates, the mediator is as good as randomly assigned, ruling out any unobserved confounders that jointly affect \( M \) and \( Y \) conditional on \( X \). This represents a selection-on-observables assumption for mediator assignment.

\begin{figure}[!htp]
	\centering
	\caption{Causal model satisfying Assumptions \ref{A1} to \ref{Acimed} }
	\label{dag_1}
	%\begin{subfigure}
		\centering
		\begin{tikzpicture}
			% DAG 1
			\node[state] (D) at (0,0) {$D$};
			\node[state] (M) [right =of D] {$M$};
			\node[state] (Y) [right =of M] {$Y$};
			\node[state] (X) [above =of M] {$X$};
			\node[state] (U1) [below =of D] {$U1$};
		  \node[state] (U2) [below =of M] {$U2$};
			\node[state] (U3) [below =of Y] {$U3$};
			% Directed edge
			\path (D) edge (M);
			\path (M) edge (Y);
			\path (X) edge (D);
			\path (X) edge (M);
			\path (X) edge (Y);
			\path[dashed] (U1) edge (D);
			\path[dashed] (U2) edge (M);
			\path[dashed] (U3) edge (Y);
		\end{tikzpicture}
\end{figure}
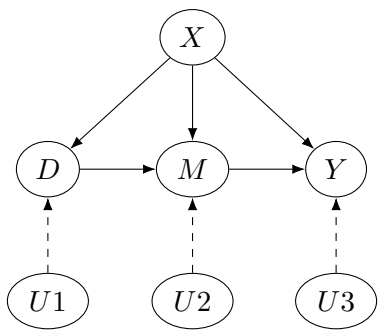

Figure \ref{dag_1} provides a causal model represented by a causal graph, where nodes denote variables (or sets of variables) and arrows indicate causal relationships between them, as discussed in \citet{Pearl00}. In addition to the previously defined variables \( Y \), \( D \), \( M \), and \( X \), the variables \( U1 \) to \( U3 \) represent unobserved factors, as shown by the dashed arrows, indicating that their effects are not directly observed. This model satisfies Assumptions \ref{A1} and Assumptions \ref{A3} to \ref{Acimed}. For example, there is no reverse causality from \( Y \) to \( M \) or \( D \), as stipulated in Assumption \ref{A1}. Furthermore, \( D \) influences \( M \) as per Assumption \ref{A3}, and it does not directly affect \( Y \) given \( X \), as required by Assumption \ref{Afullmed}. Additionally, \( D \) is as good as random conditional on \( X \), as outlined in Assumption \ref{A4}, with no unobservables jointly affecting \( D \) and \( Y \), or \( D \) and \( M \). Finally, \( M \) is as good as random conditional on \( X \), as specified in Assumption \ref{Acimed}, meaning there are no unobserved confounders jointly affecting \( M \) and \( Y \). As shown in Figure \ref{dag_1}, \( D \) is conditionally independent of \( Y \) given both \( M \) and \( X \), which is the testable implication of the joint satisfaction of these assumptions.

To discuss our testing approach more formally, consider the following conditional independence condition, which is testable in the data:
\begin{equation*}
Y \bot  D | M=m, X=x, \ \ \ \ \forall m \in \mathcal{M},\textrm{ and }  x  \in \mathcal{X}. \label{TI} \tag{TI}
\end{equation*}
Theorem \ref{theorem1} below shows that, conditional on Assumptions~\ref{A1}, \ref{A3}, and \ref{A4}, the testable implication \eqref{TI} is a necessary and sufficient condition for the joint satisfaction of Assumptions~\ref{Afullmed} and \ref{Acimed}.
% when considering potential outcomes \(Y(m)\) that correspond to the factual mediator realization \(M = m\).  This result implies that, in practice, Assumptions~\ref{Afullmed} and \ref{Acimed} can only be tested for factual (observed) outcomes, such as the potential outcomes \(Y(1)\) for units with \(M = 1\) and \(Y(0)\) for units with \(M = 0\). In contrast, it is not possible to construct tests based on counterfactual outcomes, such as \(Y(0)\) for units with \(M = 1\) or \(Y(1)\) for units with \(M = 0\). Consequently, our approach tests a necessary, but not sufficient, condition with respect to the joint distribution of potential outcomes.  If violations of Assumptions \ref{Afullmed} and \ref{Acimed} occur exclusively in counterfactual outcomes, for instance \(Y(0)\) for mediated units and \(Y(1)\) for non-mediated units, such violations cannot be detected by our testing procedure. From a practical perspective, however, it appears unlikely that violations would arise solely among counterfactual outcomes but not among factual ones, as this would require highly restrictive data-generating mechanisms. The proof of Theorem \ref{theorem1} is provided in Appendix \ref{appendixtheorem1}.

\begin{theorem}\label{theorem1}
Under Assumptions \ref{A1}, \ref{A3}, \ref{A4}: Assumptions \ref{Afullmed}, \ref{Acimed} $\iff$ (\ref{TI}).
\end{theorem}

The assumptions underlying Theorem \ref{theorem1} allow for testing independence across the entire conditional potential outcome distribution, as required when aiming to identify distributional treatment effects, such as effects at specific quantiles of the potential outcome distribution. Examples include the total treatment effects considered in \citet{Firpo03} and the causal mechanisms analyzed in \citet{hsu2023doublyrobustestimationdirect}. If the focus is on average treatment effects as the causal parameters of interest, the conditional independence assumptions \ref{A4} and \ref{Acimed} can be relaxed to conditional mean independence assumptions, while Assumption \ref{A3} is to be replaced by a conditional mean dependence assumption. In this case, an analogous theorem for testing in the context of average treatment and mediator effects can be established, in the same vein as Theorem 2 in \citet{huberkueck2022} (which we omit here for brevity). We also note that our practical implementation of the test based on DML, as outlined in Section \ref{testing}, focuses on mean parameters.

\section{Causal framework without (conditionally) random treatments}\label{Assumptions2}

In this section, we discuss the testable implications without imposing Assumption \ref{A4}, imposing that the treatment is as good as randomly assigned conditional on covariates. For this reason, we shift from experimental settings to the context of observational data, in which treatment assignment is typically not exogenous. Under Assumptions \ref{A1} to \ref{A3} alone, it can be shown that the satisfaction of the testable conditional independence \eqref{TI} implies and is implied by the joint satisfaction of Assumptions \ref{Afullmed} and \ref{Acimed}, as well as a conditional independence assumption on the treatment that is weaker than Assumption \ref{A4}. To better see this, we first note that  Assumption \ref{A4} implies the following two weaker marginal independence assumptions
\setcounter{assumption}{4}
\begin{assumptionA}[Conditional independence of the treatment and potential outcome]\label{A4a}  
  \begin{align*}  
   Y(d,m)  \bot D \mid X = x \quad \forall d \in \mathcal{D}, m \in \mathcal{M}, \text{ and } x \in \mathcal{X}.  
  \end{align*}  
\end{assumptionA}
\begin{assumptionB}[Conditional independence of the treatment and potential mediator]\label{A4b}  
  \begin{align*}  
   M(d)  \bot D \mid X = x \quad \forall d \in \mathcal{D}, m \in \mathcal{M}, \text{ and } x \in \mathcal{X}.  
  \end{align*}  
\end{assumptionB}  
Assumption \ref{A4a} states that, conditional on $X$, there are no confounders jointly affecting the treatment $D$ and the outcome $Y$ when the mediator $M$ is held fixed at some value $m$. In other words, there are no confounders that influence both the treatment and the outcome directly, other than through the mediator. It is also worth noting that, together with Assumption \ref{Afullmed} (full mediation), Assumption \ref{A4a} implies $$Y(m) \bot D | X=x \quad \forall m \in \mathcal{M}, \textrm{ and } x \in \mathcal{X}.$$ Assumption \ref{A4b} rules out confounders jointly affecting the treatment and the mediator. 

Theorem \ref{theorem2}, which closely follows Theorem 1 in \citet{huberkueck2022}, formally shows that, conditional on Assumptions \ref{A1} and \ref{A3}, the testable implication \eqref{TI} is necessary and sufficient for Assumptions \ref{Afullmed}, \ref{Acimed}, and \ref{A4a} when considering (observed) potential outcomes $Y(m)$ that correspond to the factual mediator realization $M = m$. Thus, \eqref{TI} requires that the mediator-outcome relationship is conditionally unconfounded, while confounders of the treatment and the mediator may be present, since Assumption \ref{A4b} is not required for the satisfaction of \eqref{TI}. The proof is provided in Appendix \ref{appendixtheorem2}.

\begin{theorem}\label{theorem2}
%Under Assumptions \ref{A1}, \ref{A3}: Assumptions \ref{A4a}, \ref{Afullmed}, \ref{Acimed} $\iff$  \ref{Afullmed}, \ref{Acimed}, \ref{A7}  $\iff$ (\ref{TI}).
Under Assumptions \ref{A1}, \ref{A3}: Assumptions \ref{A4a}, \ref{Afullmed}, \ref{Acimed}  $\iff$ (\ref{TI}).
\end{theorem}

Our result in Theorem \ref{theorem2} implies that full mediation and the identifiability of the effect of the mediator $M$ on the outcome $Y$ can be tested, along with the absence of unobservables affecting both $D$ and $Y$ directly (i.e., not through $M$). However, the identification of the indirect treatment effect, and thus the total treatment effect under full mediation, cannot be verified based on our testable implication, because \eqref{TI} does not provide information about the existence of confounders jointly affecting $D$ and $M$ conditional on $X$. Such confounders would imply that the first-stage effect of $D$ on $M$, as postulated in Assumption \ref{A3}, cannot be identified. As a result, the indirect (and total) effect of $D$ on $Y$ operating through $M$ cannot be identified, even if the second-stage effect of $M$ on $Y$ is identifiable. 

Such a causal scenario is depicted in Figure \ref{dag_2}. In contrast to Figure \ref{dag_1}, $U_2$ now jointly affects both $D$ and $M$, thereby introducing confounding when measuring the effect of $D$ on $M$. This confounding rules out the identification of the (indirect or total) effect of $D$ on $Y$ through $M$. At the same time, $D$ remains conditionally independent of $Y$ given both $M$ and $X$, ensuring that the testable implication still holds.
\begin{figure}[!htp]
	\centering
	\caption{Causal model with treatment-mediator confounding }
	\label{dag_2}
	%\begin{subfigure}
		\centering
		\begin{tikzpicture}
			% DAG 1
			\node[state] (D) at (0,0) {$D$};
			\node[state] (M) [right =of D] {$M$};
			\node[state] (Y) [right =of M] {$Y$};
			\node[state] (X) [above =of M] {$X$};
			\node[state] (U1) [below =of D] {$U1$};
		  \node[state] (U2) [below =of M] {$U2$};
			\node[state] (U3) [below =of Y] {$U3$};
			% Directed edge
			\path (D) edge (M);
			\path (M) edge (Y);
			\path (X) edge (D);
			\path (X) edge (M);
			\path (X) edge (Y);
			\path[dashed] (U1) edge (D);
            \path[dashed] (U2) edge (D);
			\path[dashed] (U2) edge (M);
			\path[dashed] (U3) edge (Y);
		\end{tikzpicture}
\end{figure}
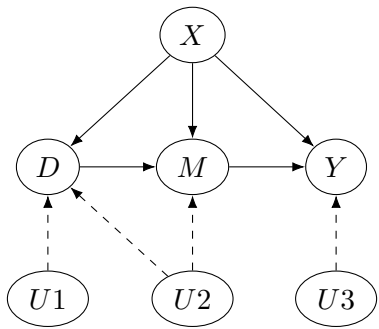

%\begin{theorem}\label{theorem2}
%Conditional on Assumptions \ref{A1} to \ref{A3}, it holds that:
%\begin{eqnarray}\label{mainresult2}
%\Pr(Y(d,m) = Y(m) | X=x) = 1, \quad Y(m) \bot M | X=x, \quad Y(m) \bot D | X=x \notag \\
%\iff \quad Y \bot D | M=m, X=x \quad \forall d \in \mathcal{D}, m \in \mathcal{M}, \textrm{ and } x \in \mathcal{X}.
%\end{eqnarray}
%Conditional on Assumptions \ref{A1} to \ref{A3}, the testable implication \( Y \bot D | M=m, X=x \) is both necessary and sufficient for the joint satisfaction of Assumptions \ref{Afullmed}, \ref{Acimed}, and \ref{A7}, when considering potential outcomes \( Y(m) \) matching the factual mediator assignment \( M=m \). The proof is given in Appendix \ref{appendixtheorem2}.
%\end{theorem}

%In many empirical applications, it is of interest to test the validity of Assumptions \ref{Afullmed}, \ref{Acimed}, and \ref{A4} (rather than Assumptions \ref{Afullmed}, \ref{Acimed}, and \ref{A4a}, as considered in Theorem \ref{theorem2}), conditional on Assumptions \ref{A1} to \ref{A3}. Such a test would also assess the identifiability of the total treatment effect of $D$ on $Y$, which, under our assumptions, coincides with the indirect effect operating through $M$. This is not possible based on the previously derived testable implication. 

\section{Comparisons with the back-door and front-door criterion}\label{Assumptions3}

This section discusses the relation between our testable implication \eqref{TI} and two identification strategies aiming at identifying total treatment effects or related causal parameters given covariates, such as conditional potential outcome distributions or conditional average treatment effects (CATE). Using the terminology of \citet{Pearl00}, the first strategy relies on the back-door (BD) criterion by controlling for $X$. This approach permits identification of conditional causal parameters under conditional independence of the treatment and potential post-treatment variables, as imposed in Assumption \ref{A4}, which is commonly referred to as selection-on-observables or unconfoundedness; see, for instance, the discussion in \citet{Im04}. The second strategy applies the front-door (FD) criterion conditional on $X$, which yields identification under conditional independence of the treatment and the potential mediator, as postulated in Assumption \ref{A4b}, conditional full mediation as postulated in Assumption \ref{Afullmed}, and conditional independence of the mediator and potential outcomes given covariates $X$ and treatment $D$; see, e.g., \cite*{Pearl95}, \cite*{Fulcheretal2019}, and \cite*{Bellemare2024}.

A natural question is whether one can learn something about the joint satisfaction of the identifying assumptions by verifying whether the identification results obtained from the BD and FD approaches coincide, for instance by means of a \citet{Hausman78}-type overidentification test, and how such an approach compares to our testable implication \eqref{TI}. To discuss such a testing approach more formally, consider the conditional probability that a discrete potential outcome takes a specific value $y$ under treatment $d$, $\Pr(Y(d,M(d))=y \mid X=x)$. Under the selection-on-observables assumption, the BD approach identifies this quantity by the conditional probability appearing on the left-hand side of the equals sign in equation \eqref{FDBD} below. In contrast, the conditional FD approach aims at computing $\Pr(Y(d,M(d))=y \mid X=x)$ by the expression appearing on the right-hand side of the equals sign in equation \eqref{FDBD}, for a discrete treatment, mediator, and outcome. An overidentification test may therefore be constructed by verifying whether the parameter values obtained from both approaches are equivalent, that is, whether it holds that
\begin{align*}
&\Pr(Y=y|D=d, X=x)=  \label{FDBD} \tag{BD=FD} \\ 
&\sum_m \Pr(M=m|D=d,X=x)\cdot \left(\sum_{d'} \Pr(Y=y|D=d',M=m,X=x)\cdot \Pr(D=d'|X=x)\right)\\
&\forall y \in \mathcal{Y}, d \in \mathcal{D}, \textrm{ and }  x  \in \mathcal{X}.
\end{align*}

Relating this approach to our previous testable implication, it turns out that the satisfaction of \eqref{TI} implies \eqref{FDBD}, such that the conditional parameters given $X$ based on FD and BD coincide at the population level, which is formally stated in Theorem \ref{theorem3}. However, this immediately implies that a comparison of BD and FD, although it appears to be an overidentification test, does not allow one to test whether all assumptions required for identification under BD and FD hold jointly. The reason is that the null hypothesis of FD and BD yielding identical parameter values holds even if Assumption \ref{A4b} is violated, see Theorem \ref{theorem2} in the previous section. Such a violation introduces bias in the measurement of conditional potential outcomes under both FD and BD that is equal in magnitude for the two methods. Proof of Theorem \ref{theorem3} is provided in Appendix~\ref{appendixtheorem3}.

\begin{theorem}\label{theorem3}
 \eqref{TI} $\Rightarrow$ \eqref{FDBD}.%  when considering potential outcomes $Y(m)$ matching the factual mediator assignment $M=m$.
 % This one is only about observable quantities, so disclaimer is not neeeded.
\end{theorem} 

While \eqref{TI} implies \eqref{FDBD}, the opposite does not hold in general. BD and FD may be equivalent at the population level even though \eqref{TI} fails. This implies that equivalence of FD and BD is not sufficient for inferring full mediation (whereas \eqref{TI} is sufficient by Theorem \ref{theorem2}). The reason is that \eqref{FDBD} imposes an equality condition on weighted averages over the mediator distribution (see the term $\sum_m \Pr(M=m \mid D=d, X=x)$), which may hold due to cancellation effects across values of $m$, even if the conditional distributions $\Pr(Y \mid D=d, M=m, X=x)$ differ across values of $d$.

For equivalence of \eqref{FDBD} to also imply \eqref{TI}, an additional assumption is required. One example of such a sufficient condition is a parametric assumption imposing additive separability of the effects of the treatment and the mediator on the outcome, conditional on covariates. Specifically, this assumption requires that the conditional distribution (or mean) of the outcome given the treatment, mediator, and covariates can be decomposed into a component depending on the treatment and covariates, and a separate component depending on the mediator and covariates. This rules out interaction effects between the treatment and the mediator on the outcome, which substantially restricts effect heterogeneity. Such no-interaction assumptions have been considered in the mediation analysis literature \citep{RoGr92,Robins2003} to achieve identification of direct and indirect effects without imposing (cross-world) independence assumptions between treatments and counterfactual (on top of factual) outcomes \citep{andrews2021insights}, in particular in the presence of post-treatment confounders.

%completeness assumption. Completeness rules out cancellation effects by requiring that the set of conditional distributions of the mediator given the treatment and covariates is sufficiently rich to uniquely identify functions of the mediator from their conditional expectations. In other words, completeness ensures that equality of weighted averages over $M$ implies equality of the underlying functions of $M$. Formally, completeness is stated in the following assumption:
\setcounter{assumption}{6}
\begin{assumption}[Separability]\label{Asepar}
There exist two functions $\alpha: \mathcal{D} \times \mathcal{X} \rightarrow \mathbb{R}$ and $\beta: \mathcal{M} \times \mathcal{X} \rightarrow \mathbb{R},$ such that for all $y \in \mathcal{Y}, m \in \mathcal{M}, d \in \mathcal{D}, x \in \mathcal{X}$:
$$\Pr(Y=y|D=d,M=m,X=x) = \alpha(y,d,x) + \beta(y,m,x).$$

%For any measurable function $h:\mathcal{M}\to\mathbb{R}$ and for all
%$x \in \mathcal{X}$,
%$$
%\sum_m h(m)\Pr(M=m \mid D=d, X=x)=0 \quad \forall d\in\mathcal{D}\quad \textrm{ implies }\quad h(m)=0 \quad \forall %m\in\mathcal{M}.
%$$
\end{assumption}

%\begin{assumption}[Completeness]\label{Asepar}
%For any measurable function $h:\mathcal{M}\to\mathbb{R}$:
%$$\sum_{m} h(m)\Pr(M=m \mid D=d, X=x)=0 \quad \forall d\in\mathcal{D}$$
%implies
%$$h(m)=0 \quad \forall m\in\mathcal{M}, \textrm{almost surely in } X=x.$$
%\end{assumption}

Theorem \ref{theorem4} shows that, conditional on separability, \eqref{TI} is both necessary and sufficient for equivalence of the BD and FD approaches. As a consequence, under Assumption \ref{Asepar}, the assumptions underlying Theorems \ref{theorem1} and \ref{theorem2} can also be tested based on a comparison of FD and BD. The proof is provided in Appendix~\ref{appendixtheorem4}.
\begin{theorem}\label{theorem4}
Under Assumption \ref{Asepar}:
\eqref{TI} $\iff$ \eqref{FDBD}.
%when considering potential outcomes $Y(m)$ matching the factual mediator assignment $M=m$.
\end{theorem}

Under the separability assumption \ref{Asepar}, it follows that the $\alpha(y,d,x)$ component of the conditional probability 
$\Pr(Y=y \mid D=d, M=m, X=x)$ does not depend on the treatment value $d$, as shown in the proof of Theorem \ref{theorem4} 
in Appendix~\ref{appendixtheorem4}. This essentially implies the absence of a direct effect. 
The empirical plausibility of separability strongly depends on the specific application. For instance, if the magnitude of a biomarker's mediated effect depends on the administered drug dose (i.e., the treatment), this constitutes a treatment-mediator interaction, and the separability assumption is violated. We note that several mediation studies have proposed relaxations of, or alternatives to, the no-interaction assumption; see, for example, \citet{ImYa2011}, \citet{TchetgenTchetgenVanderWeele2012}, and  references provided in the survey by \citet{huber2020mediation}.

\section{Testing}\label{testing}

This section presents the empirical implementation of our testing approach, focusing on a version of the testable implication \eqref{TI} that is appropriate when evaluating (conditional) average treatment effects. Accordingly, we consider conditional mean independence rather than full independence as the relevant object for testing. This choice is motivated by the fact that most empirical applications of mediation analysis focus on average (direct and indirect) effects, for which conditional mean independence is the binding constraint. Testing full conditional independence would require nonparametric distributional tests that are considerably more demanding in terms of sample size and dimensionality, while conditional mean independence can be tested efficiently using the doubly robust DML framework below. We propose a test based on double machine learning (DML) for the following null hypothesis $H_0$, which is equivalent to conditional mean independence between the treatment $D$ and the outcome $Y$, given the mediator $M$ and covariates $X$:
\begin{align}\label{nullhyp1} 
H_0: E[Y \mid M = m, X = x, D = d] - E[Y \mid M = m, X = x] = 0 \quad \forall m \in \mathcal{M},\ x \in \mathcal{X},\ d \in \mathcal{D}.
\end{align}

We apply the DR testing approach proposed by \citet{Apfel2023learning}, originally developed to jointly test the validity of an instrument and the exogeneity of a treatment conditional on covariates. In our application, the treatment corresponds to the instrument in their setting, and the mediator corresponds to their treatment. For a binary treatment \( D \), we define the propensity score as \( p(M, X) = \Pr(D = 1 \mid M, X) \). Furthermore, let \( \mu(M, X, D) = E[Y \mid M, X, D] \) denote the conditional mean outcome. The DR score function for the binary treatment case is
\begin{align}\label{score3}  
&\quad \tilde{\psi}(W,\theta,\eta)\\
&=(\mu(M,X,1)-\mu(M,X,0))^2\notag\\&+2(\mu(M,X,1)-\mu(M,X,0))\left(\frac{(Y-\mu(M,X,1))\cdot D}{p(M,X)}
-\frac{(Y-\mu(M,X,0))\cdot (1-D)}{1-p(M,X)}\right)\notag\\
&+\mu(M,X,1)-\mu(M,X,0)+\left(\frac{(Y-\mu(M,X,1))\cdot D}{p(M,X)}
-\frac{(Y-\mu(M,X,0))\cdot (1-D)}{1-p(M,X)}\right)-\theta,\notag
\end{align}
where the target parameter $\tilde\theta$ is  
\begin{align}
\tilde\theta = E\left[\left(\mu(M,X,1)-\mu(M,X,0)\right)^2\right] + E\left[\mu(M,X,1)-\mu(M,X,0)\right].
\end{align}
Testing the null hypothesis \eqref{nullhyp1} based on \eqref{score3} corresponds to an aggregate \( L_2 \)-type measure, commonly used in specification tests relying on nonparametric regression \citep{racine1997consistent, racine2006testing, hong1995consistent, wooldridge1992test}. 

\citet{Apfel2023learning} also extend \eqref{score3} to multivalued treatments \( D \), which requires discretizing \( D \) when it is continuous. Let \( \{D_l\}_{l=1}^L \) be a partition of \( \mathcal{D} \) such that \( \bigcup_l D_l = \mathcal{D} \). For discrete \( D \), we define \( D_l = d_l \) for those values with \( \Pr(D = d_l) > c \) for some constant \( c > 0 \), while for continuous \( D \), the partition can be defined using quantiles. Let \( 1(D \in D_l) \) be the indicator function, and define the treatment propensity score as \( p_l(M, X) = \Pr(D \in D_l \mid M, X) \). The DR score function for multivalued \( D \) is  
\begin{align}\label{score4}
&\quad \psi(W,\theta,\eta)\\
&=\sum_{l=1}^L (\mu(M,X,D\in D_l)-\mu(M,X,D\notin D_l))^2\notag\\&
+\sum_{l=1}^L 2 (\mu(M,X,D\in D_l)-\mu(M,X,D\notin D_l))\notag\\
&\quad\left(\frac{(Y-\mu(M,X,D\in D_l)) 1(D\in D_l)}{p_l(M,X)}
-\frac{(Y-\mu(M,X,D\notin D_l))1(D\notin D_l)}{1-p_l(M,X)}\right)\notag\\&
+\sum_{l=1}^L (\mu(M,X,D\in D_l)-\mu(M,X,D\notin D_l))\notag\\&
+\sum_{l=1}^L \frac{(Y-\mu(M,X,D\in D_l)) 1(D\in D_l)}{p_l(M,X)}
-\frac{(Y-\mu(M,X,D\notin D_l)) 1(D\notin D_l)}{1-p_l(M,X)} - \theta\notag,
\end{align}
where the target parameter $\theta$ is  
\begin{align}
  \theta =   E\left[\sum_{l=1}^L(\mu(M,X,D\in D_l)-\mu(M,X,D\notin D_l))^2+(\mu(M,X,D\in D_l)-\mu(M,X,D\notin D_l))\right].
\end{align}

\citet{Apfel2023learning} show that both score functions, \( \tilde{\psi} \) and \( \psi \) have expectation zero under the null and satisfy Neyman orthogonality. This implies that a sample analog of \( \theta \) is \( \sqrt{n} \)-consistent - where \( n \) denotes the sample size - and asymptotically normal under specific regularity conditions. A key requirement is that the product of the estimation errors for \( \mu(M, X, D) \) and \( p(M, X) \) vanishes at rate \( o(n^{-1/2}) \) as the sample size grows. This condition is met if the estimation error of both models is of order \( o(n^{-1/4}) \). Another important condition is that independent subsamples are used for estimating \( \mu(M, X, D) \) and \( p(M, X) \), on the one hand, and for estimating the score function \( \tilde{\psi} \), on the other hand. This is obtained through cross-fitting, i.e., by iteratively swapping the subsamples used for the estimation tasks, as described in \citet{Chetal2018}.

%For comparing ATE under selection-on-observables and front-door, use DR procedures - see papers for frontdoor:
%\cite*{Fulcheretal2019} and \cite*{Gorbachetal2023} discuss doubly robust (DR) method for the indirect effect/ATE estimation.
%Concerning the comparison of the ATE based on the conditional front-door estimate and the ATE based on the selection-on-observables assumption (which is yet another test). %The conditional DR front-door estimator of \cite*\cite*{Gorbachetal2023} is available in R code at tetianagorbach/semiparametric_inference_ACE_BD_FD_TD_efficiency_robustness. For selection-on-observables take the treatDML command of the causalweight package. 
%Now, ideally, you would compute a score that is the difference between the DR score functions of front-door and selection-on-observables for the test. Then the mean of that difference score function would be test statistic and the variance would be useful for inference 

Finally, we present an analogous DML approach for testing differences in conditional mean outcomes across the FD and BD approaches (even though it appears less attractive given the caveats discussed in Section \ref{Assumptions3}), which corresponds to a conditional mean version of \eqref{FDBD}. To this end, we denote by $q(d,x)=E[Y|D=d, X=x]$ the conditional mean outcome given the treatment and the covariates. This is the parameter considered by the BD approach and corresponds to the conditional mean potential outcome $E[Y(d,M(d)) \mid X=x]$ if the identifying assumptions underlying BD, most importantly a selection-on-observables assumption such as Assumption \ref{A4}, are satisfied. Concerning the FD approach, we denote by $\nu(m,x)$ the nested conditional mean outcome given $X=x$, which is obtained by averaging the conditional mean outcome $\mu(m,x,d)$ over treatment values $D$ while holding the mediator fixed at value $m$:
\begin{eqnarray}
\nu(m,x)=\sum_{d} \mu(m,x,d) f (d|X=x)=E[ \mu(m,x,D) | X=x],
\end{eqnarray}
where $f(d \mid X)=\Pr(D=d \mid X)$ denotes the treatment propensity score, and the treatment and mediator are assumed to be discrete. 

As discussed, for instance, in \citet{frsp2019}, the conditional FD criterion corresponds to the conditional mean of $\nu(m,x)$ given the treatment and the covariates, $E[ \nu(M,x) |D=d,X=x]$, which equals the conditional mean potential outcome $E[Y(d,M(d)) \mid X=x]$ if the identifying assumptions underlying the FD approach are satisfied. \cite*{Fulcheretal2019} also present an alternative representation for this parameter, namely,
\begin{eqnarray}\label{frontdoormean2}
\sum_{m} \mu(m,x,D) f(m|D=d,X=x),
\end{eqnarray}
where $f(m|D,X)=\Pr(M=m|D,X)$ denotes the conditional probability of the mediator taking value $m$ (the mediator propensity score). It is worth noting that in equation \eqref{frontdoormean2}, $D$ in $\mu(m,x,D)$ is evaluated at observed treatment values while $D=d$ is fixed in the mediator propensity score $f(m|D=d,X=x)$. For ease of notation, we will henceforth denote the expression in \eqref{frontdoormean2} by $\zeta(d,x)$. 

We aim to test for differences in the conditional mean parameters obtained under the BD and FD approaches and therefore consider the following null hypothesis, which states that their difference is zero:
\begin{align}\label{nullhyp2}
H_0: q(d,x) - \zeta(d,x) = 0. \quad \forall x \in \mathcal{X},\ d \in \mathcal{D}.
\end{align}
 Applying the approach of \citet{Apfel2023learning} to testing \eqref{nullhyp2}, we consider as target parameter
\begin{align}
  \bar\theta =   E\left[  (q(D,X)-  \zeta(D,X))^2 \right]+ E\left[  q(D,X)-  \zeta(D,X) \right].
\end{align}
Considering a discrete treatment, we note that the DR score function of $q(d,X)$ is given by the following expression, see e.g.  \cite*{Robins+94} and \cite*{RoRo95}:
\begin{align}\label{DRselobs}
q(d,X)+\frac{(Y-q(D,X))\cdot I(D = d)}{f(d \mid X)},
\end{align}
Furthermore, the DR function of the FD expression  $\zeta(d,X)$ is (for discrete $D$ and $M$) given by the following expression, see \cite*{Fulcheretal2019} and \cite*{Gorbachetal2023}:
\begin{align}\label{DRfrontdoor}
&\zeta(d,X) + [Y -  \mu(M,X,D)]\cdot \frac{f(M \mid d, X) }{  f(M \mid D, X)} \\
&+ \frac{I(D = d)}{f(D \mid X)} 
\left[
    \underbrace{\sum_{d} \mu( M,X,d) f (d|X)}_{\nu(M,X)} - \sum_{d,m} \mu( m ,X,d) f(m|D,X) f (d|X)
\right] .\notag
\end{align}

Analogously to the score functions in \eqref{score3} and \eqref{score4}, the score function for the difference between the DR score functions of the BD and FD approaches is given by
\begin{align}
\label{score5}  
\quad \bar{\psi}(W,\theta,\eta)&=(q(d,X)-\zeta(d,X))^2\notag\\&+2(q(d,X)-\zeta(d,X))\left( r(d,X)
-s(d,M,X)\right)\\
&+q(d,X)-\zeta(d,X)+\left(r(d,X)-s(d,M,X)\right)-\bar\theta,\notag
\end{align}
where
\begin{align}\label{DRselobssupp}
r(d,X)=\frac{(Y-q(d,X))\cdot I\{D=d\}}{f(d \mid X)}
\end{align}
and
\begin{align}\label{DRfrontdoorsupp}
s(d,M,X)&=[Y -  \mu(M,X,D)]\cdot \frac{f(M \mid d, X) }{  f(M \mid D, X)} \\
&+ \frac{I(D = d)}{f(D \mid X)} 
\left[
    \sum_{d} \mu( M,X,d) f (d|X) - \sum_{d,m} \mu( m ,X,d) f(m|D,X) f (d|X)
\right].\notag
\end{align}

\section{Simulations}\label{simulations}

This section evaluates the finite-sample performance of our proposed testing approach through a simulation study. We assess our method under varying levels of confounding, direct treatment effects, and sample sizes. Our first simulation design replicates the causal structure described in Section \ref{Assumptions1}, specified as follows:
\begin{eqnarray}\label{sim1}
	D &=& I\{X'\beta + U_D>0\},\notag \\
	M &=& 0.5 D +  X'\beta  + \delta U_{MY} + U_M,\notag \\
	Y &=&   M + X'\beta + \gamma D + \delta U_{MY}  + U_Y,\notag \\
	X &\sim& \mathcal{N}(0,\Sigma_X), \quad U_D, U_M, U_Y, U_{MY} \stackrel{\mathrm{iid}}{\sim} \mathcal{N}(0,1), \notag
\end{eqnarray}
with $X$, $U_D$, $U_M$, $U_Y$, and $U_{MY}$ being mutually independent.

The binary treatment variable $D$ is generated as an indicator function $I\{\cdot\}$ based on a linear combination of pre-treatment covariates $X$  - for a nonzero coefficient vector $\beta$ - and an unobserved error term $U_D$ that affects $D$ alone. The mediator $M$ is specified as a linear function of $D$ and $X$, together with an idiosyncratic shock $U_M$ and a shared unobservable $U_{MY}$ whose loading $\delta$ governs mediator-outcome confounding. The outcome variable $Y$ is a linear function of mediator $M$, covariates $X$ (if $\beta \neq 0$), treatment $D$ (if $\gamma \neq 0$, implying a direct effect), an idiosyncratic shock $U_Y$, and the same shared unobservable $U_{MY}$ (with loading $\delta$). When $\gamma = 0$, the treatment effect is fully mediated by $M$, satisfying Assumption \ref{Afullmed}. When $\gamma \neq 0$, a direct effect of $D$ on $Y$ exists, violating full mediation. The unobserved terms $U_D, U_M, U_Y, U_{MY}$ are standard normally distributed and independent of each other and of $X$, which is a vector of $p$ normally distributed covariates with zero means and covariance matrix $\Sigma_X$, where the $(i,j)$th element of $\Sigma_X$ corresponds to $0.5^{|i-j|}$. The coefficient vector $\beta$ gauges the effects of the covariates on $Y$, $M$, and $D$, quantifying the degree of confounding due to observables. The $i$th element of $\beta$ is set to $0.5/i^2$ for $i=1,\ldots,p$, implying a quadratic decay in the relevance of any additional covariate $i$ for confounding. Crucially, in design \eqref{sim1} the confounder $U_{MY}$ enters only $M$ and $Y$, so $\delta$ isolates mediator-outcome confounding without spilling into $D$; this ensures that Assumption \ref{A4} continues to hold for all values of $(\delta,\gamma)$ used below.

To facilitate the interpretation of the violation parameters, we anchor each of them to a standard sensitivity-analysis quantity. For the confounding loading $\delta$ (and analogously for $\lambda$ and $\theta$ in design \eqref{sim2}), the value $\rho$ corresponds to a partial $R^2$ of the unobserved confounder in the affected equation, conditional on the other regressors, equal to $\rho^2 / (1 + \rho^2)$. This is the quantity used as the basis of the sensitivity-analysis framework of \citet{cinelli2020}. For the direct-effect parameter $\gamma$, the corresponding interpretable quantity is the fraction of the total average treatment effect that bypasses the mediator: $\gamma / (0.5 + \gamma)$, since under design \eqref{sim1} the total ATE equals $0.5 + \gamma$. The values $\delta = 0.5$ and $\gamma = 0.125$ used in our main simulations both correspond to a violation strength of $20\%$ on their respective natural scales: a confounder explaining $20\%$ of the residual variation of the affected variable, or a direct effect accounting for $20\%$ of the total ATE. This anchoring places our chosen magnitudes within the range of empirical applications considered by \citet{cinelli2020} in their flagship benchmarking exercise.

Our study comprises 1000 simulations for each of two distinct sample sizes \( n \) consisting of 1000 and 4000 observations respectively, with the number of covariates \( p \) set to 200. To test the null hypothesis \eqref{nullhyp1}, we estimate the conditional means $\mu(M,X,1)$ and $\mu(M,X,0)$ as well as the propensity score $p(M,X)$ involved in the score function \eqref{score3} using lasso regression via the \texttt{SuperLearner} package in \textsf{R}. Other machine learning methods, such as random forests, gradient boosting, or ensemble learners, could alternatively be employed, provided they satisfy the convergence rate requirements outlined in Section \ref{testing}. We utilize 5-fold cross-fitting and trim observations with propensity scores below 0.05 or above 0.95 to prevent inflation of the inverse probability weights.

The parameters $\delta$ and $\gamma$ allow us to systematically introduce violations of the mediator exogeneity assumption (Assumption \ref{Acimed}) and full mediation assumption (Assumption \ref{Afullmed}), respectively. Setting $\delta = 0$ and $\gamma = 0$ implies satisfaction of both assumptions, such that the testable implication \eqref{TI} holds. When $\delta \neq 0$, the shared unobservable $U_{MY}$ becomes a confounder jointly affecting $M$ and $Y$, violating Assumption \ref{Acimed} while leaving Assumption \ref{A4} intact. When $\gamma \neq 0$, the treatment exhibits a direct effect on the outcome not operating through $M$, violating Assumption \ref{Afullmed}. Table \ref{tab:simulations_test1} presents the results of our simulations under various choices of $\delta$ and $\gamma$, reporting the average estimate $\hat{\theta}$, the standard deviation of the estimates across simulations (std.\ $\hat{\theta}$), the average standard error (mean SE), and the test's rejection rate (rej.\ rate) at the 5\% level of statistical significance.

\begin{table}[ht!]
  \centering
  \begin{adjustbox}{max width=\textwidth}
  \begin{threeparttable}
  \caption{Simulations: Test for Full Mediation and Mediator Exogeneity}\label{tab:simulations_test1}
  
  \begin{tabular}{c|c c c c}
  \midrule
  \midrule
  sample size & $\hat{\theta}$ & std.\ $\hat{\theta}$ & mean SE & rej.\ rate \tabularnewline
  \midrule
  \multicolumn{5}{c}{$\delta = 0$ \& $\gamma = 0$}\tabularnewline
  \midrule
  1000 & $-$0.003 & 0.091 & 0.095 & 0.046 \tabularnewline
  4000 & $-$0.003 & 0.040 & 0.041 & 0.050 \tabularnewline
  \midrule
    \multicolumn{5}{c}{$\delta = 0.5$ \& $\gamma = 0$ \quad ($20\%$ partial $R^2$, M-Y confounding)}\tabularnewline
    \midrule
    1000 & $-$0.116 & 0.083 & 0.091 & 0.282 \tabularnewline
    4000 & $-$0.097 & 0.035 & 0.038 & 0.736 \tabularnewline
  \midrule
  \multicolumn{5}{c}{$\delta = 0$ \& $\gamma = 0.125$ \quad ($20\%$ of total ATE is direct)}\tabularnewline
  \midrule
  1000 & 0.139 & 0.107 & 0.111 & 0.208 \tabularnewline
  4000 & 0.140 & 0.047 & 0.050 & 0.835 \tabularnewline
  \midrule
  \midrule
  \end{tabular}
  \begin{tablenotes}
    \small
    \item Notes: Column `$\hat{\theta}$' gives the average estimate of the target parameter across simulations. `std.\ $\hat{\theta}$' gives the standard deviation of estimates across simulations. `mean SE' gives the average standard error. `rej.\ rate' gives the empirical rejection rate when setting the level of statistical significance to 0.05.
  \end{tablenotes}
  \end{threeparttable}
  \end{adjustbox}
\end{table}

The top panel of Table \ref{tab:simulations_test1} shows the results for $\delta = 0$ and $\gamma = 0$, ensuring satisfaction of both full mediation (Assumption \ref{Afullmed}) and mediator exogeneity (Assumption \ref{Acimed}). The test maintains close adherence to the nominal 5\% significance level, with rejection rates of 4.6\% at $n = 1000$ and 5.0\% at $n = 4000$. The average estimate $\hat\theta$ is essentially zero at both sample sizes, and the standard deviation of estimates across simulations is well-approximated by the average standard error (e.g.\ 0.040 vs.\ 0.041 at $n = 4000$), indicating that the asymptotic variance formula performs adequately even at $n = 1000$.

The intermediate panel reports the results for $\delta = 0.5$ and $\gamma = 0$, corresponding to a violation of mediator exogeneity (Assumption \ref{Acimed}) of magnitude $20\%$ partial $R^2$, while full mediation holds. The rejection rate reaches 28.2\% at $n = 1000$ and 73.6\% at $n = 4000$, showing that the test has modest power against $M$-$Y$ confounding of this magnitude at the smaller sample size but high power once $n$ is sufficient. The negative value of $\hat\theta$ reflects collider bias induced by conditioning on $M$: when $D \perp U_{MY}$ unconditionally but both $D$ and $U_{MY}$ enter the equation for $M$, fixing $M$ creates a negative association between $D$ and $U_{MY}$, which in turn lowers the conditional mean of $Y$ given $D = 1$ relative to $D = 0$.

The bottom panel presents results for $\delta = 0$ and $\gamma = 0.125$, where mediator exogeneity holds but full mediation is violated through a direct treatment effect that accounts for $20\%$ of the total ATE. The rejection rate is 20.8\% at $n = 1000$ and 83.5\% at $n = 4000$, demonstrating that the test reliably detects this magnitude of full-mediation violation once the sample is large enough. Notably, $\hat\theta \approx 0.14$ at both sample sizes, close to the theoretical target $\gamma + \gamma^2 = 0.141$ that the score function tracks under a constant direct-effect shift. Overall, the simulations confirm satisfactory size control under the null and useful power at both confounding channels by $n = 4000$.

Our second simulation design explores scenarios in which the conditional randomization of $D$ (Assumption~\ref{A4}) is relaxed in favor of the weaker pair (\ref{A4a},~\ref{A4b}), as in Theorem~\ref{theorem2}. We introduce separate, mutually independent unobserved error terms for each potentially confounded pair, with a dedicated scalar knob for each channel:
\begin{eqnarray}\label{sim2}
    D &=& I\{X'\beta + \lambda U_{DM} + \theta U_{DY} + U_D > 0\},\notag  \\
    M &=& 0.5 D + X'\beta + \lambda U_{DM} + \delta U_{MY} + U_M,\notag \\
    Y &=& M + X'\beta + \gamma D + \delta U_{MY} + \theta U_{DY} + U_Y,\notag
\end{eqnarray}
with $X$ as before and $U_{DM}, U_{DY}, U_{MY}, U_D, U_M, U_Y$ mutually independent standard normal. By construction: $\theta$ controls $D$-$Y$ confounding (Assumption~\ref{A4a}) via the shared unobservable $U_{DY}$, $\lambda$ controls $D$-$M$ confounding (Assumption~\ref{A4b}) via $U_{DM}$, $\delta$ controls $M$-$Y$ confounding (Assumption~\ref{Acimed}) via $U_{MY}$, and $\gamma$ controls the direct effect (Assumption~\ref{Afullmed}). When $\lambda \neq 0$, $U_{DM}$ jointly affects $D$ and $M$, corresponding to the causal structure depicted in Figure~\ref{dag_2}, where identification of the indirect effect is no longer possible due to the violation of Assumption~\ref{A4b}. Theorem~\ref{theorem2} establishes that the testable implication \eqref{TI} remains satisfied as long as Assumptions~\ref{A4a},~\ref{Afullmed} (full mediation), and~\ref{Acimed} (mediator exogeneity) hold, even in the presence of $D$-$M$ confounding.

Table \ref{tab:simulations_thm2} presents simulation results for $\lambda = 0.5$, anchored to the same $20\%$ partial $R^2$ benchmark as $\delta$ and $\theta$. We hold $D$-$M$ confounding fixed at this level and vary each of the channels that the test \emph{is} designed to detect: $\theta$ ($D$-$Y$ confounding), $\delta$ ($M$-$Y$ confounding), and $\gamma$ (direct effect). This comparison illustrates what the test can and cannot detect. The presence of $D$-$M$ confounding via $\lambda$ shifts the propensity score distribution and modestly reduces the effective sample size after trimming.

\begin{table}[ht!]
  \centering
  \begin{adjustbox}{max width=\textwidth}
  \begin{threeparttable}
  \caption{Simulations: Illustrating Theorem \ref{theorem2} ($\lambda = 0.5$)}\label{tab:simulations_thm2}
  
  \begin{tabular}{c|c c c c}
  \midrule
  \midrule
  sample size & $\hat{\theta}$ & std.\ $\hat{\theta}$ & mean SE & rej.\ rate \tabularnewline
  \midrule
  \multicolumn{5}{c}{$\lambda = 0.5$, $\theta = 0$, $\delta = 0$ \& $\gamma = 0$ \quad (only $D$-$M$ confounding)}\tabularnewline
  \midrule
  1000 & $-$0.010 & 0.094 & 0.098 & 0.056 \tabularnewline
  4000 & $-$0.002 & 0.041 & 0.043 & 0.049 \tabularnewline
  \midrule
  \multicolumn{5}{c}{$\lambda = 0.5$, $\theta = 0$, $\delta = 0.5$ \& $\gamma = 0$ \quad ($+$ $M$-$Y$ confounding, $20\%$ partial $R^2$)}\tabularnewline
  \midrule
  1000 & $-$0.150 & 0.081 & 0.087 & 0.449 \tabularnewline
  4000 & $-$0.135 & 0.033 & 0.036 & 0.945 \tabularnewline
  \midrule
  \multicolumn{5}{c}{$\lambda = 0.5$, $\theta = 0.5$, $\delta = 0$ \& $\gamma = 0$ \quad ($+$ $D$-$Y$ confounding, $20\%$ partial $R^2$)}\tabularnewline
  \midrule
  1000 & 0.497 & 0.157 & 0.160 & 0.915 \tabularnewline
  4000 & 0.481 & 0.073 & 0.074 & 1.000 \tabularnewline
  \midrule
  \multicolumn{5}{c}{$\lambda = 0.5$, $\theta = 0$, $\delta = 0$ \& $\gamma = 0.125$ \quad ($+$ direct effect, $20\%$ of ATE)}\tabularnewline
  \midrule
  1000 & 0.140 & 0.108 & 0.115 & 0.186 \tabularnewline
  4000 & 0.141 & 0.050 & 0.052 & 0.799 \tabularnewline
  \midrule
  \midrule
  \end{tabular}
  \begin{tablenotes}
    \small
    \item Notes: Column `$\hat{\theta}$' gives the average estimate of the target parameter across simulations. `std.\ $\hat{\theta}$' gives the standard deviation of estimates across simulations. `mean SE' gives the average standard error. `rej.\ rate' gives the empirical rejection rate when setting the level of statistical significance to 0.05.
  \end{tablenotes}
  \end{threeparttable}
  \end{adjustbox}
\end{table}

The first panel of Table~\ref{tab:simulations_thm2} presents results for $\lambda = 0.5$, $\theta = 0$, $\delta = 0$, and $\gamma = 0$, corresponding to $D$-$M$ confounding alone, without violations of full mediation or mediator exogeneity. As predicted by Theorem \ref{theorem2}, the test maintains proper size control with rejection rates of 5.6\% at $n = 1000$ and 4.9\% at $n = 4000$, indistinguishable from the corresponding nulls in Table~\ref{tab:simulations_test1} (4.6\% and 5.0\%). This confirms that the test cannot detect, and is not designed to detect, violations of Assumption~\ref{A4b}.

The second panel reports results for $\lambda = 0.5$ and $\delta = 0.5$, where both $D$-$M$ confounding and $M$-$Y$ confounding are present (both at $20\%$ partial $R^2$). The rejection rate of 44.9\% at $n = 1000$ and 94.5\% at $n = 4000$ exceeds the rate against the same $\delta$ without $D$-$M$ confounding in Table~\ref{tab:simulations_test1} (28.2\% and 73.6\%). The additional $D$-$M$ confounding strengthens the collider association between $D$ and $U_{MY}$ induced by conditioning on $M$, which amplifies the signal the test picks up.

The third panel reports results for $\lambda = 0.5$ and $\theta = 0.5$, where $D$-$M$ confounding is present alongside $D$-$Y$ confounding (a violation of Assumption~\ref{A4a}). As Theorem~\ref{theorem2} predicts, the test detects this violation reliably, with rejection rates of 91.5\% at $n = 1000$ and 100\% at $n = 4000$. The point estimate $\hat\theta \approx 0.49 \approx \theta$ at both sample sizes, reflecting that $D$-$Y$ confounding maps directly into a difference in the conditional means $E[Y \mid D = 1, M, X] - E[Y \mid D = 0, M, X]$ that the score function targets.

The fourth panel reports results for $\lambda = 0.5$ and $\gamma = 0.125$, where $D$-$M$ confounding is present alongside a direct treatment effect ($20\%$ of the total ATE). The rejection rates of 18.6\% at $n = 1000$ and 79.9\% at $n = 4000$ are slightly below those in Table~\ref{tab:simulations_test1} (20.8\% and 83.5\%), indicating that the test maintains essentially the same power against full-mediation violations regardless of whether $D$-$M$ confounding is present. These results highlight an important implication: while our test reliably detects violations of $D$-$Y$ confounding, mediator exogeneity, and full mediation, it provides no information about $D$-$M$ confounding. Consequently, non-rejection of the test implies the identifiability of the mediator effect on the outcome, but does not guarantee identification of the total or indirect treatment effect, which additionally requires Assumption~\ref{A4b} to hold.

\section{Empirical Application}\label{application}

Perinatal depression affects millions of women worldwide, with lasting consequences for both mothers and children. \citet{baranov2020} experimentally evaluate a landmark intervention - the Thinking Healthy Programme - which delivered cognitive behavioral therapy (CBT) to depressed mothers in rural Pakistan during pregnancy and early infancy. According to \citet{baranov2020}, treated women exhibited lower depression rates and, notably, greater financial empowerment seven years later: they were more likely to control household finances and work outside the home. A natural question concerns the channels through which treatment of depression may affect economic empowerment. The authors hypothesize that improved mental health may strengthens women's social support networks within the household. They document that treated women are more likely to have a grandmother present in the home and report better relationships with their husbands. It remains unclear whether these channels fully account for the treatment's impact on financial empowerment, or whether other pathways - such as direct psychological changes including increased self-efficacy - also contribute. \citet{kwon2024testing} examine this issue using their sharp null test, which relies on a partial identification approach when testing full mediation alone and is unconditional in that it does not control for covariates.

In contrast, we jointly test full mediation and mediator exogeneity while conditioning on baseline covariates, including depression severity, wealth, and education. Specifically, we investigate whether, among women with similar baseline profiles, treatment affects financial empowerment solely through the proposed mediators (the presence of a grandmother and the quality of the relationship with the husband), and whether these mediators are exogenous such that the indirect causal mechanisms are identified in this experimental study.  We implement our test using lasso for nuisance estimation with 5-fold cross-fitting, aggregated over $S=10$ independent sample splits following \citet{Chetal2018}. Our controls include baseline demographics (mother's age and its square, education, employment status, empowerment), depression severity (Hamilton and BDQ scores), general functioning (GAF), perceived social support (MSPSS), household wealth, family structure (grandmother and mother-in-law presence, father's occupation, parity, first child indicator), and parents' education, all measured prior to treatment assignment.

Table \ref{tab:baranov} reports our findings alongside those of \citet{kwon2024testing} for comparison. Panel A presents our results conditional on baseline covariates; Panel B reproduces the unconditional test results from \citet{kwon2024testing}. We reject conditional mean independence for grandmother presence ($p = 0.011$), relationship quality ($p = 0.022$), and both mediators combined ($p = 0.013$), where standard errors are robust to clustering at the Union Council level (40 clusters), the level at which treatment was randomized. These findings suggest that, conditional on baseline characteristics, at least one of the two assumptions (full mediation or mediator exogeneity) is violated. Our results diverge notably from those of \citet{kwon2024testing}, who test for violations of full mediation using moment inequality conditions on outcome distributions across mediator (compliance) types---defined by how the mediator responds to treatment---under a monotonicity assumption of the mediator with respect to the treatment. They reject full mediation through each individual mediator ($p = 0.023$ and $p = 0.028$), but fail to reject when both mediators are combined ($p = 0.654$), concluding that the two channels together may fully explain the treatment effect. 

This discrepancy may be driven by differences in what the tests evaluate and how they do so. In the partial identification-based test of \citet{kwon2024testing}, rejection implies (asymptotically) a failure of full mediation, while mediator exogeneity is not tested. Conversely, non-rejection does not confirm that full mediation holds, as violations may be sufficiently small to still satisfy the moment inequality conditions. Our test evaluates the stronger condition of the joint satisfaction of full mediation and mediator exogeneity and does not rely on partial identification. Therefore, when the test of \citet{kwon2024testing} rejects -- as in the case of individual mediators ($p = 0.023$ and $p = 0.028$) -- our test should also reject, though not necessarily vice versa, and this is indeed what we find ($p = 0.011$ and $p = 0.022$). When their test does not reject -- for the combined mediators ($p = 0.654$) -- ours may still reject, for example due to a failure of mediator exogeneity, even if full mediation holds. Our rejection in this case ($p = 0.013$) therefore suggests that while grandmother presence and relationship quality may jointly fully mediate the treatment effect, the mediator exogeneity assumption required for causal identification of the indirect effect may be violated. Alternatively, full mediation itself may be violated, but not detected by the partial identification approach of \citet{kwon2024testing}.

In sum, our rejection indicates that the causal mechanisms are not fully identified: either some mediators are missing, the observed mediators (grandmother presence and relationship quality) are not exogenous, or both. Even if the treatment effect were fully mediated through the observed variables, violations of mediator exogeneity prevent point identification of the indirect effect. Unobserved factors, such as resilience or social skills, may jointly influence both the mediators and the outcome (financial empowerment), confounding attempts to isolate indirect effects.

\begin{table}[ht!]
\centering
\caption{Test of Conditional Mean Independence: Baranov et al.\ (2020) Data}
\label{tab:baranov}
\begin{tabular}{lcccc}
\toprule
Mediator & Test Statistic & SE (cl.) & $p$-value & $N$ \\
\midrule
\multicolumn{5}{l}{\textit{Panel A: Our Test (Conditional on $X$, $S=10$)}} \\[2pt]
\quad Grandmother presence & 0.513 & 0.202 & 0.011 & 516 \\
\quad Relationship quality & 0.420 & 0.184 & 0.022 & 503 \\
\quad Both combined & 0.435 & 0.175 & 0.013 & 503 \\[4pt]
\midrule
\multicolumn{5}{l}{\textit{Panel B: Kwon \& Roth (2024) - Partial Identification Test}} \\[2pt]
\quad Grandmother presence & - & - & 0.023 & 585 \\
\quad Relationship quality & - & - & 0.028 & 568 \\
\quad Both combined & - & - & 0.654 & - \\
\bottomrule
\end{tabular}
\vspace{6pt}
\begin{minipage}{0.9\textwidth}
\footnotesize
\textit{Notes:} Panel A reports results of our conditional mean independence test with null hypothesis $H_0: Y \perp D \mid M, X$. Baseline controls include mother's age (and square), depression severity (Hamilton and BDQ scores), general functioning (GAF), perceived social support (MSPSS), employment status, empowerment, wealth index, mother's and father's education, family structure (grandmother and mother-in-law presence, father's occupation), parity, and first child indicator, all measured prior to treatment. Nuisance functions estimated via lasso with 5-fold cross-fitting, aggregated over $S=10$ independent sample splits. Standard errors are robust to clustering at the Union Council level (40 clusters), the level at which treatment was randomized. Panel B reproduces the unconditional sharp null test from \citet{kwon2024testing}.
\end{minipage}
\end{table}

Our second application examines the role of social norms in shaping women's labor market outcomes. \citet{bursztyn2020} conduct a field experiment among married men in Saudi Arabia to study how misperceived beliefs about others' attitudes toward women working affect household decisions. Participants were asked to estimate what fraction of other men in their session would support women working outside the home. Those randomized to the treatment group then learned the true (typically higher) share, correcting their underestimates. The authors document two key findings. First, treated men were more likely to sign up their wives for a job-matching service offered at the end of the session. Second, in a follow-up survey, wives of treated men were more likely to have applied for jobs outside the home. These results highlight the need to understand mechanisms, as the information treatment may affect job applications not only through job-matching sign-up, but potentially also via other channels, such as shifting husbands' attitudes or reducing perceived social stigma.

Table~\ref{tab:bursztyn} reports the results. The partial identification -- based test of \citet{kwon2024testing} suggests that full mediation through job-matching sign-up is rejected ($p = 0.02$). The authors note that at least 11 percent of ``never-takers'' - men who would not sign up their wives for the service under either treatment -- are nevertheless directly affected by the information treatment. Applying our conditional independence test to these data, while conditioning on baseline beliefs about others' views, employment status, education, and demographic characteristics, we strongly reject conditional mean independence. While a rejection of our test can be driven by a failure of either full mediation or mediator exogeneity, these results are consistent with the findings of \citet{kwon2024testing}, who reject the sharp null of full mediation alone for this application ($p=0.02$). Our stronger rejection ($p = 0.004$) may additionally reflect violations of mediator exogeneity, as both the presence of additional mechanisms and the endogeneity of the mediator may contribute to the rejection of the joint null.

\begin{table}[ht!]
\centering
\caption{Test of Conditional Mean Independence: Bursztyn et al.\ (2020) Data}
\label{tab:bursztyn}
\begin{tabular}{lcccc}
\toprule
 & Test Statistic & SE & $p$-value & $N$ \\
\midrule
\multicolumn{5}{l}{\textit{Our Test (Conditional on $X$)}} \\[2pt]
\quad Job-matching sign-up & 0.114 & 0.040 & 0.004 & 375 \\[4pt]
\midrule
\multicolumn{5}{l}{\textit{Kwon \& Roth (2024) - Partial Identification Test}} \\[2pt]
\quad Job-matching sign-up & - & - & 0.020 & - \\
\bottomrule
\end{tabular}
\vspace{6pt}
\begin{minipage}{0.9\textwidth}
\footnotesize
\textit{Notes:} The outcome is an indicator for whether the wife applied for jobs outside the home. The mediator is job-matching service sign-up. Controls include baseline beliefs (own views and perceptions of others' views on women working, semi-segregation, and minimum wage), employment status of husband and wife, education, age, number of children, network characteristics, and session fixed effects-matching the specification in \citet{bursztyn2020}. Nuisance functions estimated via lasso with 5-fold cross-fitting, aggregated over $S=10$ sample splits.
\end{minipage}
\end{table}

\section{Conclusion}\label{conclusion}

We develop a unified framework for testing full mediation and the identifiability of indirect treatment effects in causal studies. Unlike naive outcome regressions on treatment and mediators, which can produce biased inference due to post-treatment confounding, our approach uses tests of conditional independence between the treatment and outcome given mediators and covariates. Under randomized or conditionally randomized treatments, this allows for joint testing of full mediation of the treatment effect via observed mediators and the exogeneity conditions necessary for identifying these mediated effects. When the treatment is not (conditionally) randomized, our framework can still test for full mediation, but not for the identifiability of mediated effects. In such settings, the conditional independence test informs only about mediator-outcome confounding, while treatment-mediator confounding may remain unaccounted for, preventing identification of mediated effects.

We connect our testable implication to two identification strategies in observational studies, the back-door (BD) and the front-door (FD) criteria. Satisfaction of our conditional independence condition implies that BD and FD yield identical (conditional) population parameters. This implies that, in general, equivalence of the BD and FD criteria does not guarantee identification of causal effects, since our conditional independence condition requires weaker assumptions than those necessary for effect identification - more specifically, it does not account for treatment-mediator confounding if the treatment is not (conditionally) randomized. Furthermore, while the conditional independence condition implies BD-FD equivalence, the converse only holds under an additional assumption such as a separability condition the rules out treatment-mediator interactions conditional on covariates. Such a condition appears unrealistic in many empirical applications. In this context, testing based on conditional independence appears more attractive than using a BD-FD comparison.

For empirical implementation, we propose a double machine learning (DML) approach using doubly robust (DR) score functions, which accommodates high-dimensional covariates while preserving root-n consistency and asymptotic normality. We also present a simulation study that indicates good finite-sample performance. Finally, we apply our method to two randomized experiments, \citet{baranov2020} on cognitive behavioral therapy for maternal depression and \citet{bursztyn2020} on social norms and women's employment. In both settings, we reject the joint null of full mediation and mediator exogeneity.

\newpage
\setlength\baselineskip{14.0pt}
\bibliographystyle{econometrica}
\bibliography{research.bib}

\newpage
{\large \renewcommand{\theequation}{A-\arabic{equation}} %
\setcounter{equation}{0} \appendix
}
\appendix \numberwithin{equation}{section}
{\small
\section{Appendix}

Proofs of Theorems \ref{theorem1} and \ref{theorem2} follow similar mechanics as in \citet{huberkueck2022} and \citet{huber2024testing}. We translate the identifying assumptions in this paper into DAG terminology \citep{Pearl95}.

Let $G$ denote the original graph, and let $G_D$, $G_M$, and $G_{DM}$ denote the interventional graphs in which all arrows originating from $D$, $M$, or $\{D, M\}$, respectively, are removed. This follows from the Single World Intervention Graphs (SWIGs) framework of \citet{richardson2013single}. We write $A \rightsquigarrow B$ for the existence of a directed path from $A$ to $B$.

The identifying assumptions translate as follows:
\begin{itemize}
\item[(\ref{A1})] No directed paths $Y \rightsquigarrow M$, $Y \rightsquigarrow D$, $Y \rightsquigarrow X$, $M \rightsquigarrow D$, $M \rightsquigarrow X$, or $D \rightsquigarrow X$ exist in $G$ (plus causal faithfulness, which places no further restriction on $G$).
\item[(\ref{A3})] $D$ and $M$ are d-connected with conditioning set $\{X\}$ in $G$.
\item[(\ref{A4})] $\{Y, M\}$ and $D$ are d-separated with conditioning set $\{X\}$ in $G_{DM}$.
\item[(\ref{A4a})] $Y$ and $D$ are d-separated with conditioning set $\{X\}$ in $G_{DM}$.
\item[(\ref{A4b})] $M$ and $D$ are d-separated with conditioning set $\{X\}$ in $G_{D}$.
\item[(\ref{Afullmed})] There is no open directed path $D \rightarrow Y$ with conditioning set $\{X\}$, apart from the one that goes through $M$, in $G_M$.
\item[(\ref{Acimed})] $Y$ and $M$ are d-separated with conditioning set $\{X\}$ in $G_{DM}$.
\end{itemize}
The testable implication translates as:
\begin{itemize}
\item[(\ref{TI})] $Y$ and $D$ are d-separated with conditioning set $\{M, X\}$ in $G$.
\end{itemize}

Below we provide analytical proofs. We also note that these proofs were verified computationally by exhaustive search in the space of DAGs with observed variables $Y,D,M,X$ and unobserved confounders for all their pairs, see  \citet{huber2024testing} for details.

\subsection{Proof of Theorem \ref{theorem1}}\label{appendixtheorem1}

\begin{proof}

``$\implies$''

We proceed by a proof by contradiction. Suppose that $\neg(\ref{TI})$ holds (throughout the proofs, we use the notation $\neg(S)$ to denote that the statement $S$ does not hold). This means that there exists an open path $Y - \cdots - D$ in the graph $G$ conditional on $\{M, X\}$.

Assumption (\ref{A4a}) (which from Assumption (\ref{A4})) states that there is no open path $Y - \cdots - D$ in graph $G_{DM}$ conditional on $\{X\}.$

$\neg$(\ref{TI}) and (\ref{A4a}) can only hold simultaneously, if either of the following situations occur:

\begin{itemize}
\item[(A)] $M$ is a collider and it opens the otherwise closed path $Y - \cdots - D$,
\item[(B)] adding arrows from $D$ and/or $M$ (difference between graphs $G$ and $G_{DM}$) opens the otherwise closed path $Y - \cdots - D$.
\end{itemize}

(A) Suppose that $M$ is a collider and it opens the path $Y - \cdots - D$:

$M$ is a collider, so we have $Y - \cdots \rightarrow M \leftarrow \cdots - D$. Now consider the this part  $Y - \cdots \rightarrow M$ of the path, which has to be open. We therefore have one of these cases:
\begin{itemize}
\item $Y \rightarrow M$ this contradicts (\ref{A1}). \Lightning 
\item $Y \leftarrow U \rightarrow M$ this contradicts (\ref{Acimed}). \Lightning
\item $Y \leftarrow X \rightarrow M$ this path is closed, because $X$ is conditioned on. \Lightning
\item $Y \leftarrow U \rightarrow X \rightarrow M$ this path is closed, because $X$ is conditioned on. \Lightning
\item $Y \leftarrow X \leftarrow U \rightarrow  M$ this path is closed, because $X$ is conditioned on. \Lightning
\item $Y \rightarrow X \rightarrow M$ this contradicts (\ref{A1}). \Lightning
\end{itemize}

Therefore $M$ is not a collider that opens up the closed path $Y - \cdots - D$.

(B) Adding arrow(s) from $D$ and/or $M$ opens the otherwise closed path $Y - \cdots - D$

Consider the following situations when adding arrow from $D$ would open up the path $Y - \cdots - D$.

\begin{itemize}
\item $D \rightarrow Y$ this contradicts (\ref{Afullmed}). \Lightning
\item $D \rightarrow X$ this contradicts (\ref{A1}). \Lightning
\item $D \rightarrow M \leftarrow \cdots Y$ this was ruled out in part (A) of the proof. \Lightning
\item $D \rightarrow M \rightarrow \cdots Y$ which is one of the following:
\begin{itemize}
\item $D \rightarrow M \rightarrow X \cdots Y$ this contradicts (\ref{A1}). \Lightning
\item $D \rightarrow M \rightarrow Y$ which does not lead to a violation of (\ref{TI}). \Lightning
\end{itemize}

\end{itemize}

Consider the following situations when adding arrow from $M$ would open up the path $Y - \cdots - D$.
\begin{itemize}
\item $D \dots M \rightarrow  \cdots Y$ which can only result in
\begin{itemize}
\item $D \dots M \rightarrow  Y$ which does not lead to a violation of (\ref{TI}). \Lightning
\item $D \dots M \rightarrow X \cdots Y$ this contradicts (\ref{A1}). \Lightning
\end{itemize}
\item $D \dots \leftarrow M \cdots Y$ which can only result in
\begin{itemize}
\item $D \leftarrow M \cdots Y$ this contradicts (\ref{A1}). \Lightning
\item $D \dots X \leftarrow M \cdots Y$ this contradicts (\ref{A1}). \Lightning
\end{itemize}
\end{itemize}

We conclude that neither (A) nor (B) can happen and therefore ``$\implies$'' holds.

\medskip
\medskip

Part ``$\impliedby$''

\medskip

``(\ref{Afullmed}) $\impliedby$''

Assume that (\ref{A1}), (\ref{A3}), (\ref{A4a}), (\ref{A4b}), (\ref{TI}), $\neg(\ref{Afullmed})$.

$\neg(\ref{Afullmed})$ together with (\ref{A1}) means that there is a direct link $D \rightarrow Y.$ But this directly contradicts (\ref{TI}). \Lightning

\medskip

``(\ref{Acimed}) $\impliedby$''

Assume that (\ref{A1}), (\ref{A3}), (\ref{A4a}), (\ref{A4b}), (\ref{TI}), $\neg(\ref{Acimed})$.

$\neg(\ref{Acimed})$ says that there is an open path between $Y$ and $M$ in graph $G_{DM}$ conditional on $X$. Because there are no arrows from $D$ or $M$ in $G_{DM},$ this is equivalent to the existence of path $M \leftarrow U \rightarrow Y.$ By (\ref{A3}) we have the existence of $D \rightarrow M$, therefore we have $D \rightarrow M \leftarrow U \rightarrow Y.$ But the existence of this path contradicts (\ref{TI}). \Lightning

\end{proof}

\subsection{Proof of Theorem \ref{theorem2}}\label{appendixtheorem2}

\begin{proof}
``$\implies$''

This part of the proof is identical to ``$\implies$'' part of Theorem \ref{theorem1}.

Part ``$\impliedby$''

``(\ref{Afullmed}) $\impliedby$" and ``(\ref{Acimed}) $\impliedby$'' are identical to the one of Theorem \ref{theorem1}.

``(\ref{A4a}) $\impliedby$''

Assume that (\ref{A1}), (\ref{A3}), (\ref{TI}) and $\neg$(\ref{A4a}).

$\neg$(\ref{A4a}) says that there exists an open path between $Y$ and $D$ in $G_{DM}$ conditional on $X$.

By (\ref{TI}) there is no open path between $Y$ and $D$ in $G$ conditional on $\{X,M\}$. This implies that such path does not exist even in $G_{DM}.$ This implies that conditioning on $M$ closes the path that is open in $\neg$(\ref{A4a}). This means that $M$ is not a collider on this path, so at least one of the arrows have to go from $M.$ But this contradicts the fact that there are no arrows from $M$ in $G_{DM}.$ \Lightning

\end{proof}

\subsection{Proof of Theorem \ref{theorem3}}\label{appendixtheorem3}

For simplicity, we only present the proof for discrete variables.

\begin{proof}
{\footnotesize
Let $d \in \mathcal{D}$, $m \in \mathcal{M}$, $y \in \mathcal{Y}$, and $x \in \mathcal{X}$ be arbitrary. We start from the conditional independence relation in \eqref{TI} and derive (\ref{FDBD}) equality.
\begin{eqnarray*}
\Pr(Y=y|D=d, M=m, X=x) &\overbrace{=}^{(\ref{TI})}& \Pr(Y=y|D=d',M=m, X=x) \\
&\implies& \\
\sum_{d'} \Pr(Y=y|D=d, M=m, X=x) \Pr(D=d'|X=x)&=&  \sum_{d'} \Pr(Y=y|D=d',M=m, X=x)\Pr(D=d'|X=x) \\ 
 \Pr(Y=y|D=d, M=m, X=x) \underbrace{\sum_{d'}\Pr(D=d'|X=x)}_{=1}&=&  \sum_{d'} \Pr(Y=y|D=d',M=m, X=x)\Pr(D=d'|X=x) \\ 
&\implies& \\
\Pr(Y=y|D=d, M=m, X=x)&=&  \sum_{d'} \Pr(Y=y|D=d',M=m, X=x)\Pr(D=d'|X=x) \\ 
&\implies& \\
\sum_m \Pr(M=m|D=d, X=x) \cdot &=& \sum_m \Pr(M=m|D=d, X=x) \cdot \\
\Pr(Y=y|D=d,M=m, X=x) &&   \sum_{d'} \Pr(Y=y|D=d', M=m, X=x)\Pr(D=d'|X=x) \\ 
&\implies& \\
\Pr(Y=y|D=d,X=x) &\overbrace{=}^{(\ref{FDBD})}& \sum_{m} \Pr(M=m|D=d,X=x) \cdot \\
&& \sum_{d'} \Pr(Y=y|D=d', M=m,X=x)  \Pr(D=d'|X=x)
\end{eqnarray*}

The first implication follows from multipying both sides by $\Pr(D=d'|X=x)$ and summing up across $d' \in \mathcal{D}.$ Similarly, in the third implication we multiply both sides with $\Pr(M=m|D=d, X=x)$ and sum across $m \in \mathcal{M}.$ The last implication is an application of the law of total probability.

}
\end{proof}

\subsection{Proof of Theorem \ref{theorem4}}\label{appendixtheorem4}

Similarly to the proof of Theorem \ref{theorem3}, we present the proof only for discrete variables.

\begin{proof}

Let $d \in \mathcal{D}$, $m \in \mathcal{M}$, $y \in \mathcal{Y}$, and $x \in \mathcal{X}$ be arbitrary. We start from the equality of population-level quantities in FD and BD \eqref{FDBD} and derive the testable conditional independence relation in \eqref{TI}.
{\footnotesize
\begin{eqnarray*}
\Pr(Y=y|D=d,X=x) &\overbrace{=}^{(\ref{FDBD})}& \sum_{m} \Pr(M=m|D=d,X=x) \cdot \\
&& \sum_{d'} \Pr(Y=y|D=d', M=m,X=x)  \Pr(D=d'|X=x) \\
&\implies& \\
\sum_m \Pr(M=m|D=d, X=x) \cdot &=& \sum_m \Pr(M=m|D=d, X=x) \cdot \\
\Pr(Y=y|D=d,M=m, X=x) &&   \sum_{d'} \Pr(Y=y|D=d', M=m, X=x)\Pr(D=d'|X=x) \\ 
&\overbrace{\implies}^{(\ref{Asepar})}& \\
\sum_m \Pr(M=m|D=d, X=x) \cdot &=& \sum_m \Pr(M=m|D=d, X=x) \cdot \\
\left(\alpha(y,d,x) + \beta(y,m,x)\right) &&   \sum_{d'} \left(\alpha(y,d',x) + \beta(y,m,x)\right)\Pr(D=d'|X=x) \\ 
&\implies& \\
\sum_m \Pr(M=m|D=d, X=x) \alpha(y,d,x) + &=& \sum_m \Pr(M=m|D=d, X=x)\sum_{d'}\alpha(y,d',x)\Pr(D=d'|X=x) + \\
\sum_m \Pr(M=m|D=d, X=x) \beta(y,m,x)   &&   \sum_m \Pr(M=m|D=d, X=x)\beta(y,m,x) \underbrace{\sum_{d'}  \Pr(D=d'|X=x)}_{=1} \\ 
&\implies& \\
\sum_m \Pr(M=m|D=d, X=x)\alpha(y,d,x)  &=& \sum_m \Pr(M=m|D=d, X=x)\sum_{d'}\alpha(y,d',x)\Pr(D=d'|X=x) \\
&\implies& \\
\alpha(y,d,x) \underbrace{\sum_m \Pr(M=m|D=d, X=x)}_{=1} &=& \sum_{d'}\alpha(y,d',x)\Pr(D=d'|X=x)\underbrace{\sum_m \Pr(M=m|D=d, X=x)}_{=1} \\
&\implies& \\
\alpha(y,d,x)  &=& \underbrace{\sum_{d'}\alpha(y,d',x)\Pr(D=d'|X=x)}_{\textrm{does not depend on $d$}}  \\
&\implies& \\
\alpha(y,d,x)  &=& \alpha(y,d',x)\\ 
&\implies& \\
\alpha(y,d,x) + \beta(y,m,x)  &=& \alpha(y,d',x) + \beta(y,m,x)\\ 
%&\implies& \\
%\sum_{d'} \Pr(Y=y|D=d, M=m, X=x) \Pr(D=d'|X=x)&=&  \sum_{d'} \Pr(Y=y|D=d',M=m, X=x)\Pr(D=d'|X=x) \\ 
&\implies& \\
\Pr(Y=y|D=d, M=m, X=x) &\overbrace{=}^{(\ref{TI})}& \Pr(Y=y|D=d',M=m, X=x) \\
\end{eqnarray*}
}
In the first implication, we apply the law of total probability to the left-hand side. The second implication is an application of the separability assumption \ref{Asepar}. The seventh implication follows from the fact that the RHS does not depend on $d$, and therefore neither does the LHS, which, after rearrangement, leads to the condition in \eqref{TI}.
\end{proof}

\end{document}